\documentclass[conference]{sig-alternate-05-2015}
\usepackage[font=small,labelfont=bf]{caption}
\usepackage[utf8]{inputenc}
\usepackage{epstopdf}
\usepackage{graphicx}
\usepackage{mathtools}
\usepackage{multirow}
\usepackage{lipsum,subcaption}
\usepackage{array}
\usepackage{algorithm}
\usepackage{enumitem}
\usepackage{algpseudocode}
\usepackage{graphicx}
\usepackage{capt-of}
\usepackage{booktabs}
\usepackage{varwidth}
\usepackage{pbox}
\usepackage[utf8]{inputenc}
\usepackage[english]{babel}
\usepackage{amsmath,amssymb}

\usepackage{amsthm}
\usepackage{footnote}
\usepackage{makecell}
\usepackage[numbers]{natbib}
\usepackage{hyperref}
\usepackage{breakurl}
\usepackage{url}

\usepackage[bottom]{footmisc}
\newtheorem{theorem}{Theorem}
\newtheorem{problem}{Problem}

\usepackage[usenames, dvipsnames]{color}
\definecolor{orange}{rgb}{1,0.5,0}
\definecolor{magenta}{RGB}{255,0,255}
\newcommand\Mark[1]{\textsuperscript{#1}}
\title{Can Self-Censorship in News Media \\ be Detected Algorithmically?\\ A Case Study in Latin America}

\author{
\alignauthor
Rongrong Tao\Mark{1}, Baojian Zhou\Mark{2}, Feng Chen\Mark{2}, Naifeng Liu\Mark{2}, David Mares\Mark{3}, \and Patrick Butler\Mark{1}, Naren Ramakrishnan\Mark{1}
\and
\affaddr{\Mark{1} Discovery Analytics Center, Department of Computer Science, Virginia Tech, Arlington, VA, USA}\\
\affaddr{\Mark{2} Department of Computer Science, University at Albany, SUNY, Albany, NY, USA}\\
\affaddr{\Mark{3} University of California at San Diego, San Diego, CA, USA}\\
\affaddr{rrtao@vt.edu, \{bzhou6, fchen5, nliu3\}@albany.edu, dmares@ucsd.edu}\\
\affaddr{pabutler@vt.edu, naren@cs.vt.edu}
}

\makeatletter
\def\@copyrightspace{\relax}
\makeatother

\begin{document}
\maketitle

\section*{Abstract}
Censorship in social media has been well studied and provides insight into how governments stifle 
freedom of expression online. Comparatively less (or no) attention has been paid to detecting (self) censorship in traditional media (e.g., news) using social media as a bellweather. We present a novel unsupervised approach that views social media as a sensor to detect censorship in news media wherein
statistically significant differences between information published in the news media and the correlated information published in social media are automatically identified as candidate censored events. We develop a hypothesis testing framework to
identify and evaluate censored clusters of keywords, and a new near-linear-time algorithm
(called \textsc{GraphDPD}) to identify the highest scoring clusters as  indicators of censorship.
We outline 
extensive experiments on semi-synthetic data as
well as real datasets (with Twitter and local
news media) from Mexico and Venezuela, highlighting
the capability to accurately detect real-world
self censorship events.

\section{Introduction}
News media censorship is generally defined as a restriction on freedom of speech
to prohibit access to public information, and is taking place more than ever before. 
According to the Freedom of the Press Report, 40.4 percent of nations fit into the ``free'' category in 2003. By 2014, this global percentage fell to 32 percent
, as shown in Figure 1 \footnote{https://freedomhouse.org/report/freedom-press/freedom-press-2014}. More than 200 journalists were jailed in 2014, according to the Committee to Protect Journalists. In fact, in the past three years, more than 200 journalists have been jailed annually \footnote{http://saccityexpress.com/defending-freedom-of-speech/\#sthash.cbI7lWbw.dpbs}.

One of the responses to this stifling environmental context is self-censorship, i.e.,
the act of deciding not to
publish about certain topics, owing to safety or partisan
reasons.
Although the social and political aspects of news media censorship have been deeply discussed and analyzed in the field of social sciences~\citep{alkazemi2015kuwaiti, cook2013long, seib2016beyond, robinson2013pockets}, there is currently no efficient and effective approach to automatically detect and track self-censorship events in real time.

\begin{figure}[ht!]
\tiny
\centering
\includegraphics[scale=0.4]{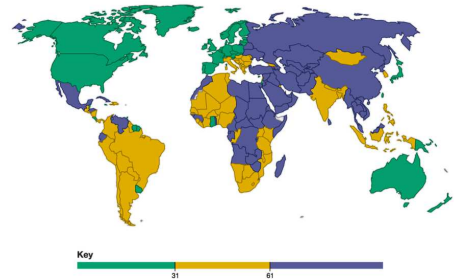}
\caption{\scriptsize Worldwide freedom of the press (2014). 
The higher the score, the worse the press freedom status. }
\label{fig:freedom}
\end{figure}

Social media censorship often takes the form of
active censors identifying offending posts and
deleting them and therefore tracking post deletions
supports the use of supervised learning
approaches~\cite{bamman2012censorship, casilli2012social, chin2014censorship,  king2013censorship}. On the
other hand, censorship in news media typically has
no labeled information and must rely on
unsupervised techniques instead.

In this paper, we present a novel unsupervised approach that views social media as a sensor to detect censorship in news media wherein
statistically significant differences between information published in the news media and the correlated information published in social media are automatically identified as candidate censored events. 

A generalized log-likelihood ratio test (GLRT) statistic can then be formulated for hypothesis testing, and 
the problem of censorship detection can  be cast as the maximization of the GLRT statistic over all possible clusters of keywords. 
We propose \underline{\smash{a near-linear-time algorithm called \textsc{GraphDPD}}} to identify the highest scoring clusters as  indicators of censorship events in the local news media, and further apply randomization testing to estimate the statistical significances of these clusters. 

We consider the detection of censorship in the news media of two countries, Mexico and Venezuela, and utilize Twitter as the uncensored source.

\subsection{Is Twitter a reliable sensor for detecting
censorship?}
Starting in January 2012, a ``Country-Withheld Content'' policy has been launched by Twitter, with which governments are able to request withholding and deletion of user accounts and tweets. At the same time, Twitter started to release a transparency report, which provided worldwide information about such removal requests.
The Transparency Report lists information and removal requests from Year 2012 to 2015 on a half-year basis. Table \ref{table:transparency_report} summarizes the information and removal requests for Year 2014 on nine countries of interest. As shown in Table \ref{table:transparency_report}, for our countries
of interest (viz. Mexico and Venezuela), 
Twitter did not participate in
any social media censorship; therefore, we
believe that Twitter can be considered as a reliable and uncensored source to detect news self censorship events in these two countries. 

The main contributions of this paper are summarized as follows:
\begin{itemize}[leftmargin=*]
\item \textbf{Analysis of censorship patterns between news media and Twitter}: We carried out an extensive analysis of information in Twitter deemed relevant to censored information in news media. In doing so, we make important observations that highlight the importance of our work. 
    \item \textbf{Formulation of an unsupervised censorship detection framework}: We propose a novel hypothesis-testing-based statistical framework for detecting clusters of co-occurred keywords that demonstrate statistically significant differences between the information published in news media and the correlated information published in a uncensored source (e.g., Twitter). 
    To the best of our knowledge, this is the first unsupervised framework for automatic detection of censorship events in news media.  
   
    \item \textbf{Optimization algorithms}: The inference of our proposed framework involves the maximization of a GLRT statistic function over all clusters of co-occurred keywords, which is hard to solve in general. We propose a novel approximation algorithm to solve this problem in nearly linear time. 
    \item \textbf{Extensive experiments to validate the proposed techniques}: We conduct comprehensive experiments on real-world Twitter and local news articles datasets to evaluate our proposed approach. The results demonstrate that our proposed approach outperforms existing techniques in the accuracy of censorship detection. In addition, we perform case studies on the censorship patterns detected by our proposed approach and analyze the reasons behind censorship from real-world data of Mexico and Venezuela during Year 2014. 
\end{itemize}

\begin{small}
\begin{table}
\tiny
\centering
\caption{\scriptsize Summary of Twitter Transparency Report for Year 2014 on nine countries of interest}
\label{table:transparency_report}
\begin{tabular}{| m{1.8cm} | m{1.4cm} | m{1.3cm} | m{1.3cm} |}
\hline
Country       & Account Information Request        & Removal Requests         & Tweets Withheld     \\ \hline
Australia     & 12  & 0   & 0   \\ \hline
Brazil     & 127  & 35   & 101    \\ \hline
Colombia     & 8  & 0   & 0   \\ \hline
Greece     & 19  & 0   & 0   \\ \hline
Japan     & 480  & 6   & 43  \\ \hline
Mexico    & 12    & 0    & 0 \\ \hline
Saudi Arabia  & 220   & 0   & 0   \\ \hline
Turkey   & 380   & 393   & 2003  \\ \hline
Venezuela &4   & 0   & 0 \\ \hline
\end{tabular}
\end{table}
\end{small}

\section{Related Work}
Here is a brief survey of three broad classes of work pertinent to our work. \\ \indent
\textbf{Analysis of the coverage of various topics across social media and news media} has been well established in many studies. 
\citep{petrovic2013can} studies topic and timing overlapping in newswire and Twitter and concludes that Twitter covers not only topics reported by news media during the same time period, but also minor topics ignored by news media. Through analysis of hundreds of news events, \citep{olteanu2015comparing} observes both similarities and differences of coverage of events between social media and news media. In this paper, we uncover indicators of censorship pattern in news media from various interactive patterns between social media and news media.
\\ \indent
\textbf{Event detection} in social media has been studied in many recent works. Watanabe et al. \citep{watanabe2011jasmine} develop a system, which identifies tweets posted closely in time and location and determine whether they are mentions of the same event by co-occurring keywords. Ritter et al. \citep{ritter2012open} presents the first open-domain system for event extraction and an approach to classify extracted events based on latent variable models. Rozenshtein et al. \citep{rozenshtein2014event} formulates event detection in activity networks as a graph mining problem and proposes effective greedy approaches to solve this problem. In addition to textual information, Gao et al. \citep{gao2015multimedia} propose an event detection method which utilizes visual content and intrinsic correlation in social media. \\ \indent
\textbf{Censorship} is a critical problem in many countries across the world and most of the existing studies on censorship analysis are focused on Turkey and China. 
Turkey, which is identified as the country issuing the largest number of censorship requests by Twitter, has been studied for censorship topics by applying topic extraction and clustering on a collection of censored tweets in \citep{tanash2015known}. \citep{coskuntuncel2016privatization} analyzes the relationship between the Turkish government and media companies and reveals that the government exerts control over mainstream media and the flow of information. However, most of the existing approaches are supervised or semi-supervised, which rely on collections of censored posts, and highlight the necessity of unsupervised approaches to uncover self censorship in news media.

\section{Data Analysis}
\label{section:data analysis}
Table \ref{table:parameters} summarizes the notation used in this work. The EMBERS project \citep{ramakrishnan2014beating} provided a collection of Latin American news articles and Twitter posts. The news dataset was sourced from around 6000 news agencies during  2014 across the world. From ``4 International Media \& Newspapers'', we retrieved a list of top newspapers with their domain names in the target country. \textit{News} articles are filtered based on the domain names in the URL links. \textit{Twitter} data was collected by randomly sampling 10\% (by volume) tweets from January 1, 2014 to December 31, 2014. Retweets in \textit{Twitter} were removed as they were not as informative as original tweets. Mexico and Venezuela were chosen as two target countries in this work since they had no censorship in Twitter (as shown in Table \ref{table:transparency_report}) but featured severe censorship in news media (as shown in Fig. \ref{fig:freedom}).
\\ \indent
\subsection{Data Preprocessing}
\label{section:data preprocessing}
The inputs to our proposed approach are keyword co-occurrence graphs. 
Each node represents a keyword associated with four attributes: (1) time-series daily
frequency (TSDF) in \textit{Twitter}, (2) TSDF in \textit{News}, (3) expected daily frequency in \textit{Twitter}, and (4) expected daily frequency in \textit{News}. Each edge represents the co-occurrence of connecting nodes in \textit{Twitter}, or \textit{News}, or both. However, constructing such graphs is not trivial due to data integration. One challenge is to handle the different vocabularies used in \textit{Twitter} and \textit{News}, with underlying distinct
distributions. \\ \indent

To find words that behave differently in \textit{News} comparing to \textit{Twitter}, we only retained keywords which are mentioned in both \textit{Twitter} and \textit{News}. For each keyword, linear correlation between its TSDF in \textit{Twitter} and \textit{News} during Year 2014 is required to be greater than a predefined threshold (e.g. 0.15) in order to guarantee the keyword is well correlated in two data sources. TSDF in \textit{Twitter} and \textit{News} for each node are normalized with quantile normalization. An edge is removed if its weight is less than $\Gamma$, where $\Gamma$ is the threshold used to tradeoff graph sparsity and connectivity. Empirically we found $\Gamma = 10$ to be an effective threshold. A keyword co-occurrence graph for a continuous time window is defined as the maximal connected component from a union of daily keyword co-occurrence graph during the time window.

\subsection{Pattern Analysis}
\label{section:pattern-analysis}
Though many events drive both social media and traditional news media, it's challenging to claim that any deviation between the two is evidence of censorship or different topics of interest. Table \ref{table:categories} summarizes various co-occurring patterns between Twitter and news media that we are able to observe from our real world dataset in Mexico and more details are discussed as follows. \indent

\begin{table*}[h!]
\tiny
     \begin{center}
    \caption{\scriptsize Different patterns of co-occurrence observed between social media and news media sources.}
     \begin{tabular}{ c  c  c   c   c  }
     \toprule
     \pbox{3cm}{Topic is of interest in both social media and news media.}
      &
      \pbox{3cm}{Topic is of interest in social media but not in news media.}
      &
      \pbox{3cm}{Topic is of interest in news media but not in social media.}
      &
      \pbox{3cm}{Censorship in one news media source.}
      &
         \\
    \cmidrule(r){1-1}\cmidrule(lr){2-2}\cmidrule(l){3-3}\cmidrule(l){4-4}
     \raisebox{-\totalheight}{\includegraphics[width=0.15\textwidth, height=20mm]{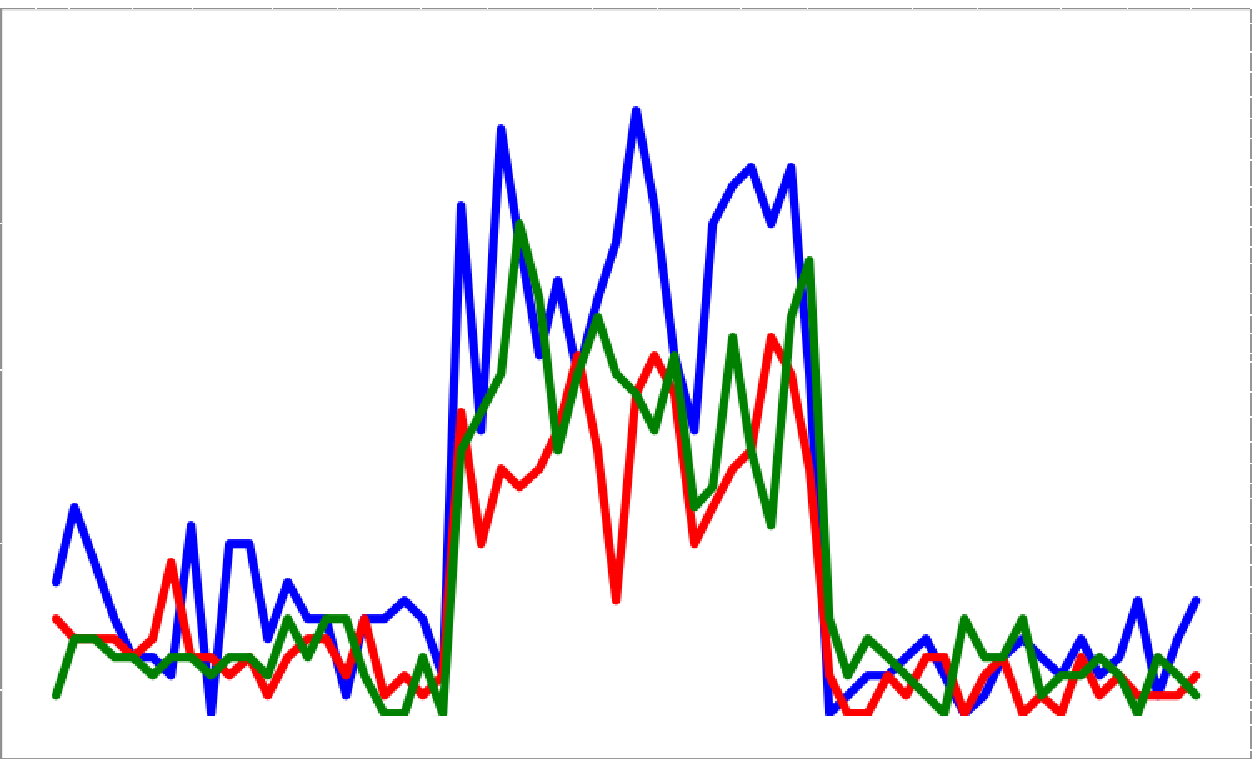}}
      & 
      \raisebox{-\totalheight}{\includegraphics[width=0.15\textwidth, height=20mm]{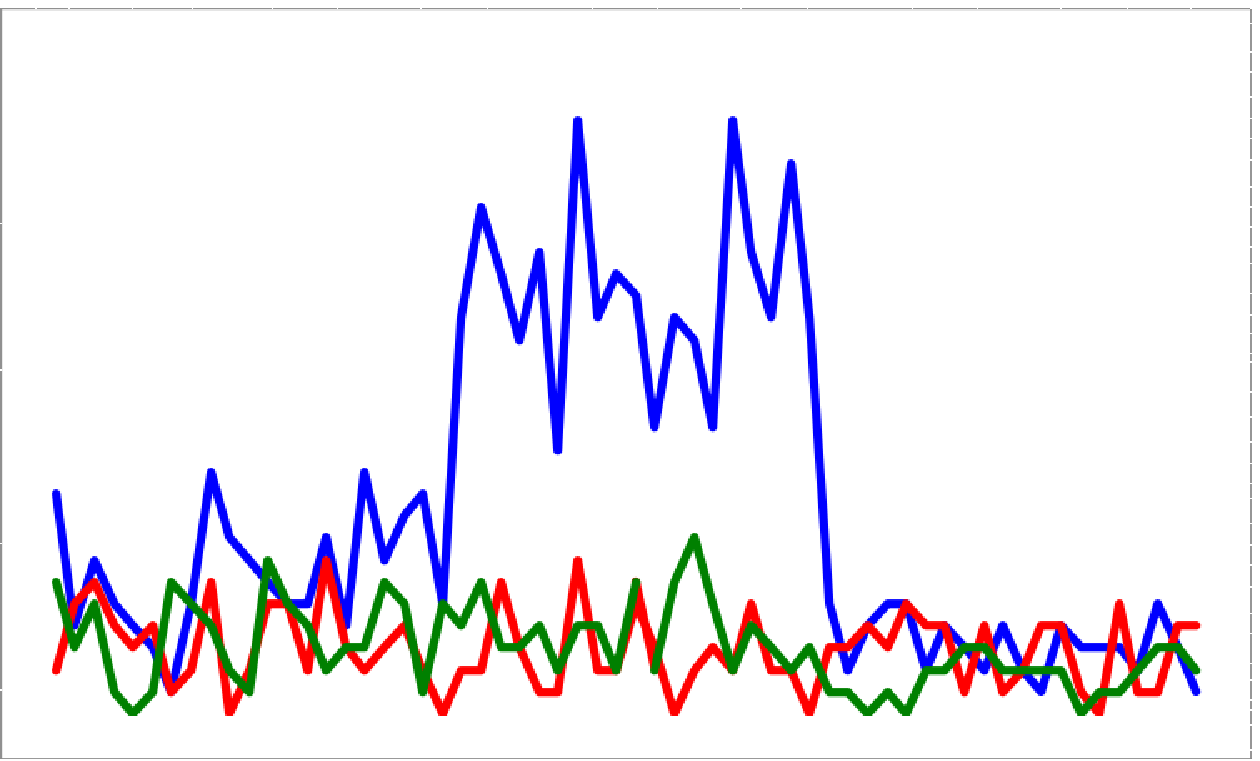}}
      & 
      \raisebox{-\totalheight}{\includegraphics[width=0.15\textwidth, height=20mm]{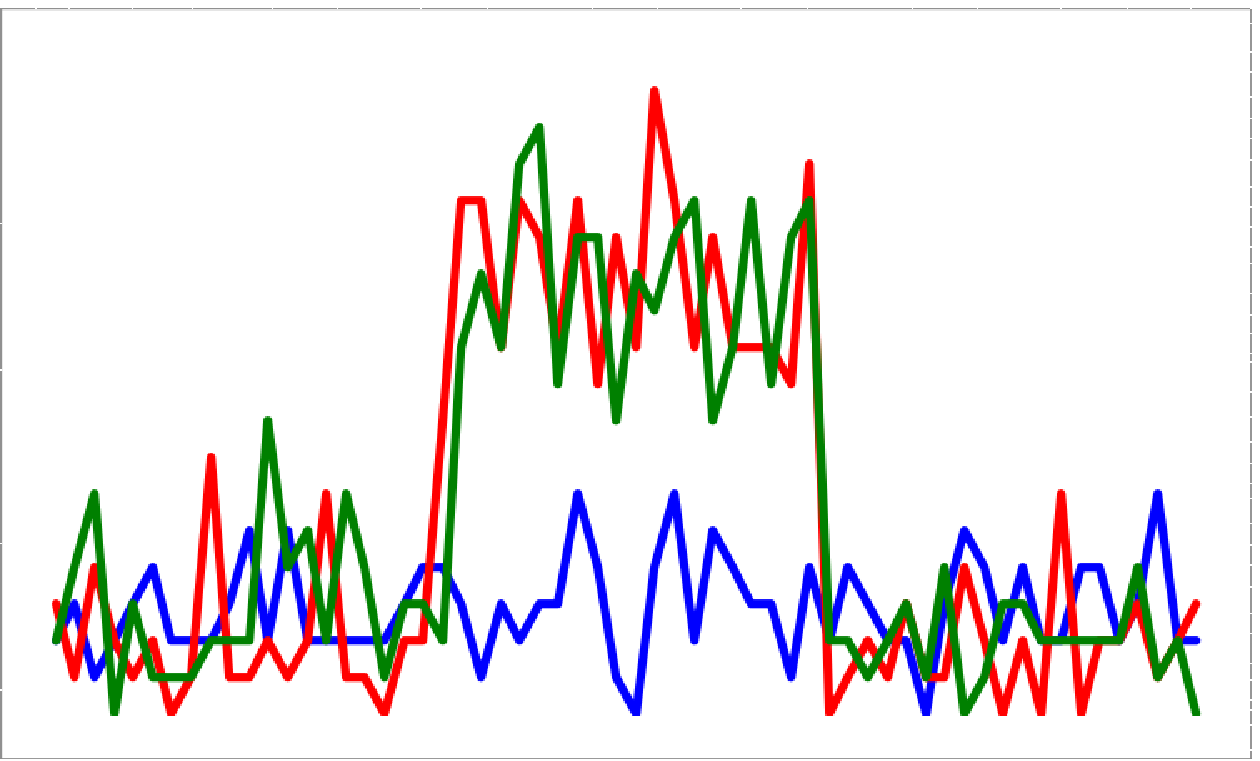}}
      &
      \raisebox{-\totalheight}{\includegraphics[width=0.15\textwidth, height=20mm]{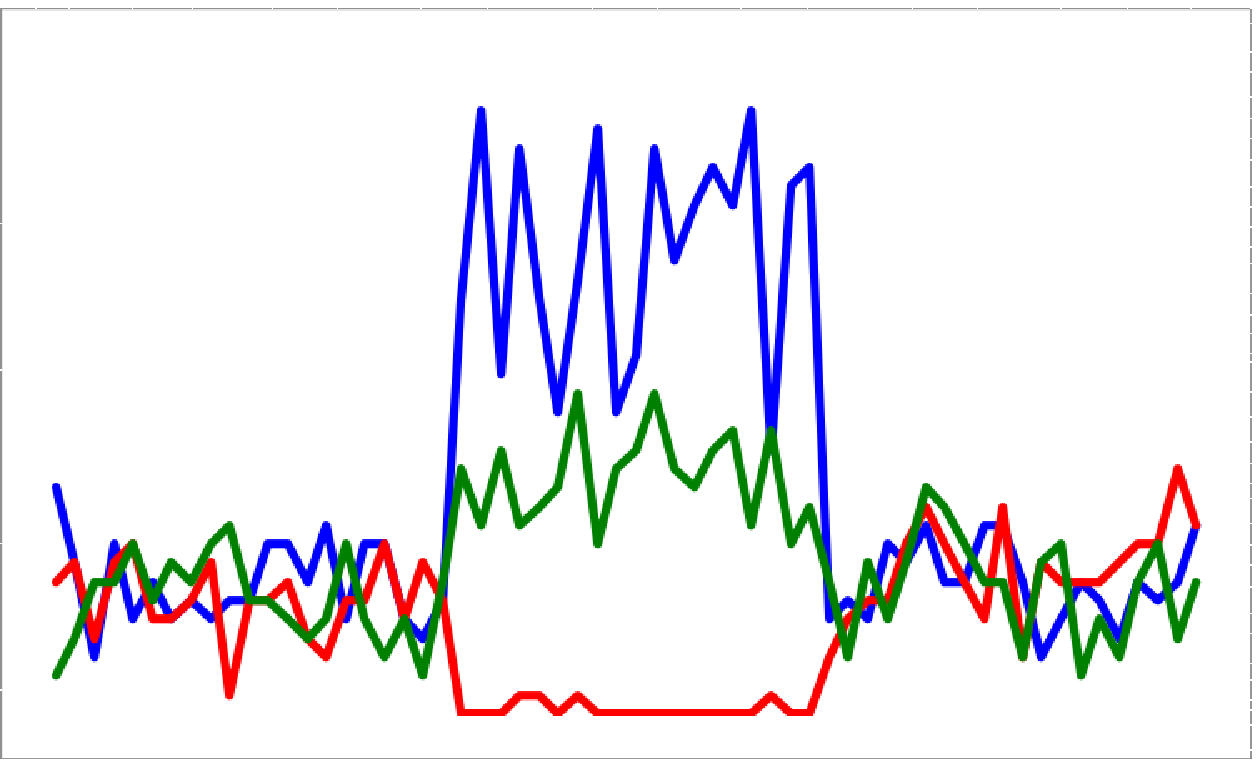}}
      &
      \raisebox{-\totalheight}{\includegraphics[width=0.15\textwidth, height=20mm]{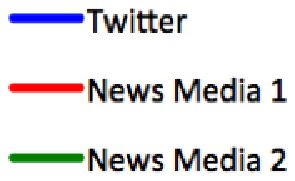}} \\
    \cmidrule(r){1-1}\cmidrule(lr){2-2}\cmidrule(l){3-3}\cmidrule(l){4-4}
    \pbox{3cm}{Example: In early March 2014, Malaysia Airlines Flight MH 370 went missing.}
    &
    \pbox{3cm}{Example: Late June 2014 featured a soccer game between Mexico and Holland as part of the 2014 FIFA World Cup.}
    &
    \pbox{4cm}{Example: In late September 2014, 125 heads of state and governments attended the Global Climate Summit, which was seen as a milestone to a new legal agreement on climate change.}
    &
    \pbox{4cm}{Example: In late September 2014, 43 students from Ayotzinapa Rural Teachers' College went missing in Mexico. This incident has been referred to as the worst human rights crisis Mexico faced since the 1968 massacre of students.}
      
      \\ \bottomrule
      \end{tabular}
      \label{table:categories}
      \end{center}
      \end{table*}

\textbf{Topic is of interest both in social media and news media}: On March 8th, 2014, Malaysia Airlines Flight MH370 disappeared while flying from Malaysia to China;  12 Malaysian crew members and 227 passengers from 15 nations were declared missing. During the following week, we are able to observe sparks in discussions of this incident and mentions of relevant keywords (MH370, Malaysia) across both social media and news media. \\ \indent
\textbf{Topic is of interest only in social media}: From June 28th to 30th 2014, there are many soccer matches held by the 2014 FIFA World Cup, including one game between Mexico and Holland. During this time period, we are able to observe spikes in mentions of relevant keywords (fifa, fútbol, robben, holland, mexicano) across Twitter in Mexico. However, mentions of these keywords in a list of nine Mexican news outlets do not depict significant changes as this is viewed as a general soccer game. \\ \indent
\textbf{Topic is of interest only in news media}: On September 23, 2014, 125 heads of state and governments attended the global Climate Summit, which was seen as a milestone to a new legal agreement on climate change. This incident is widely discussed in news media, while relatively less attention in social media (in Latin
America). \\ \indent
\textbf{Topic is censored in news media}: To illustrate an example anomalous behavior in \textit{News}, Fig. \ref{fig:introexample} compares TSDF in El Mexicano Gran Diario Regional (el-mexicano.com.mx) and TSDF in \textit{Twitter} during a 2-month period on a connected set of keywords sampled from tweets from
Mexico. All the example keywords are relevant to the 43 missing students from Ayotzinapa in the city of Iguala protesting the government’s education reforms. The strong connectivity of these keywords, as shown in Fig. \ref{fig:2e}, guarantees that they are mentioned together frequently in Twitter and local news media. The time region during which anomalous behavior is detected is highlighted with two yellow markers. Since volume of \textit{Twitter} is much larger than volume of \textit{News}, TSDF in Fig. \ref{fig:2a} to Fig. \ref{fig:2d} are normalized to [0, 500] for visualization. Fig. \ref{fig:2a} to Fig. \ref{fig:2d} depict that TSDF in El Mexicano Gran Diario Regional is well correlated with TSDF in \textit{Twitter} except during the highlighted time region, where abnormal absenteeism in El Mexicano Gran Diario Regional can be observed for all example keywords. In order to validate the deviation between TSDF in El Mexicano Gran Diario Regional and TSDF in \textit{Twitter} is not due to difference in topics of interests, we also compare with a number of other local news outlets. Fig. \ref{fig:2a} to Fig. \ref{fig:2d} shows that TSDF in El Universal in Mexico City is consistent with TSDF in \textit{Twitter} and does not depict an abnormal absenteeism during the highlighted time period. Using Twitter and El Universal in Mexico City as sensors, we can conclude an indicator of self-censorship in  El Mexicano Gran Diario Regional with respect to the 43 missing students during the highlighted time region.

Inspired by these observations, we say that a \textbf{censorship pattern} exists if for a cluster of connected keywords, 
\begin{enumerate}
    \item their TSDF in at least one local news media is consistently different from TSDF in \textit{Twitter} during a time period,  
    \item their TSDF in local news media are consistently well correlated to TSDF in \textit{Twitter} before the time period, and
    \item their TSDF in at least one different local news outlet does not depict abnormal absenteeism during the time period.
    \end{enumerate}

\begin{small}
\begin{table}[!ht]
\tiny
\centering
\caption{\scriptsize Description of major notation.}
\label{table:parameters}
\begin{tabular}{| m{1.5cm} | m{5.5cm} |}
\hline
Variable       &  Meaning        \\ \hline
$\{a^t(v)\}_{t=1}^T$      & time series of daily frequency of node $v$ in uncensored Twitter dataset     \\ \hline
$\lambda_a(v)$     & expected daily frequency of node $v$ in the Twitter dataset.       \\ \hline
$\{b^t(v)\}_{t=1}^T$     & time series of daily frequency of node $v$ in the censored news dataset   \\ \hline
$\lambda_b(v)$ & expected daily frequency of node $v$ in data source $b$  \\ \hline
TSDF   & time series of daily frequency     \\ \hline
\end{tabular}
\end{table}
\end{small}

\begin{figure}[!t]
    \centering
    \begin{subfigure}[b]{0.22\textwidth}
            \centering
            \includegraphics[width=\textwidth]{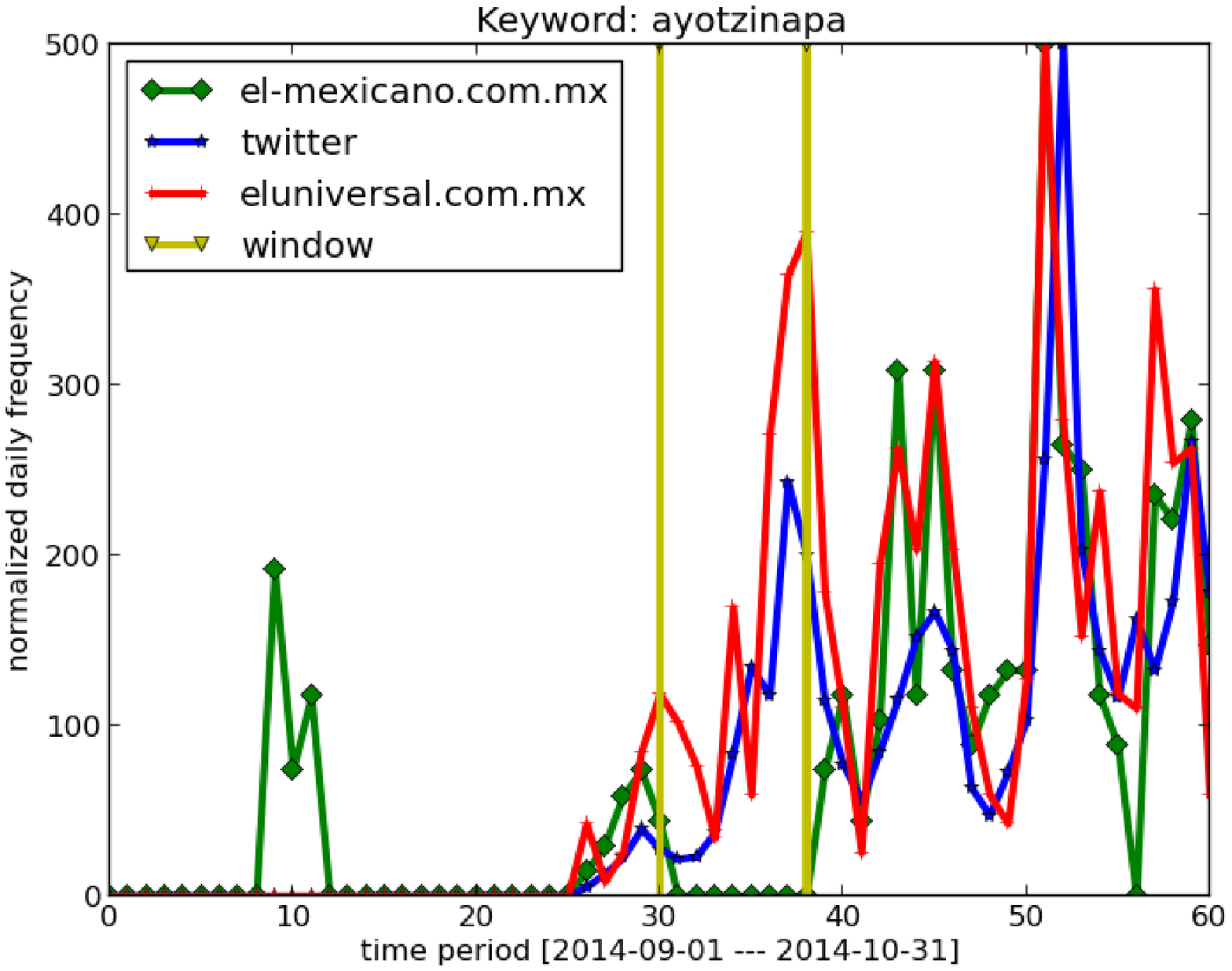}
            \caption{\scriptsize ayotzinapa}
    \label{fig:2a}
    \end{subfigure}
	\begin{subfigure}[b]{0.22\textwidth}
            \centering
            \includegraphics[width=\textwidth]{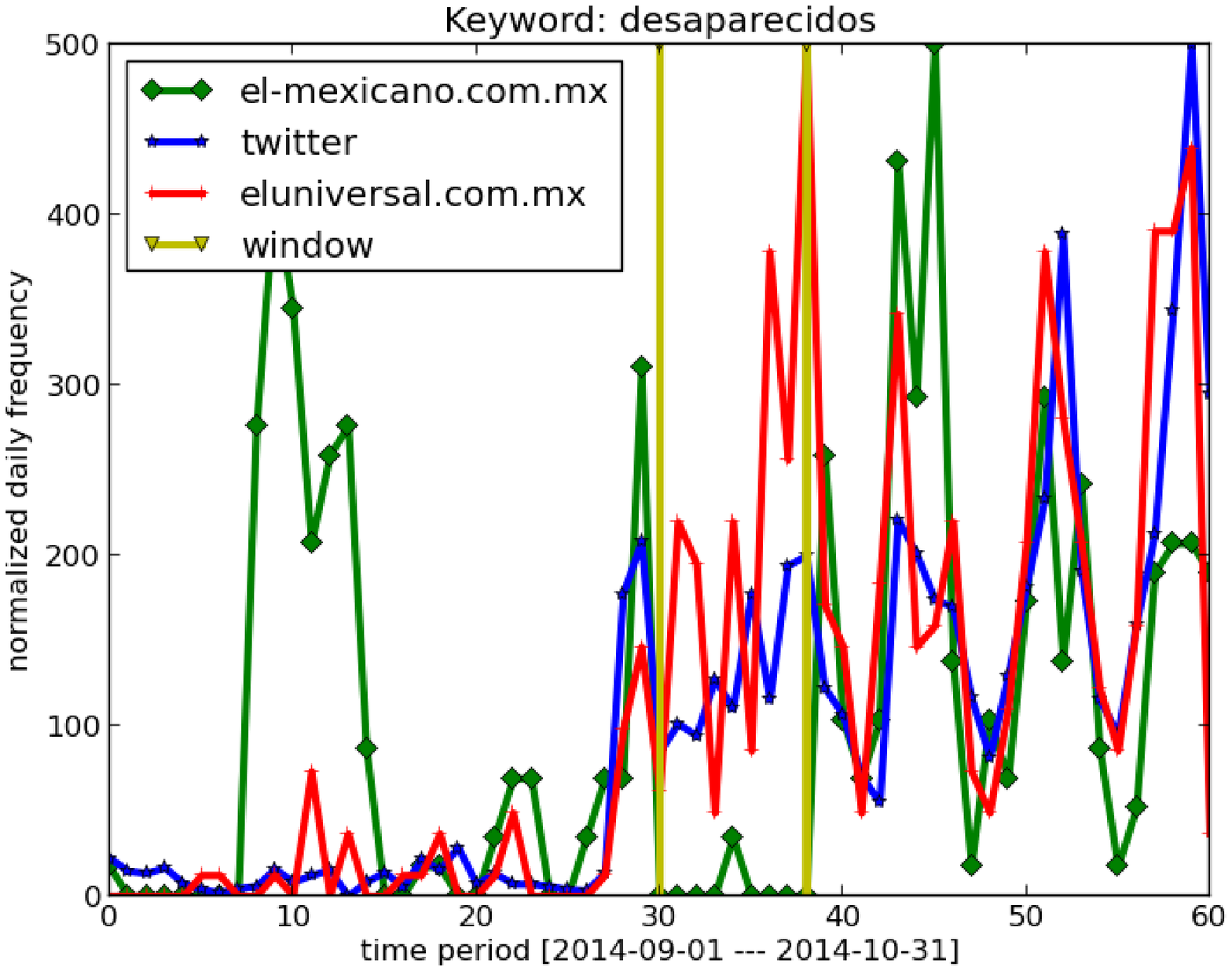}
            \caption{\scriptsize desaparecidos (missing)}
    \label{fig:2b}
    \end{subfigure}
    \begin{subfigure}[b]{0.22\textwidth}
            \centering
            \includegraphics[width=\textwidth]{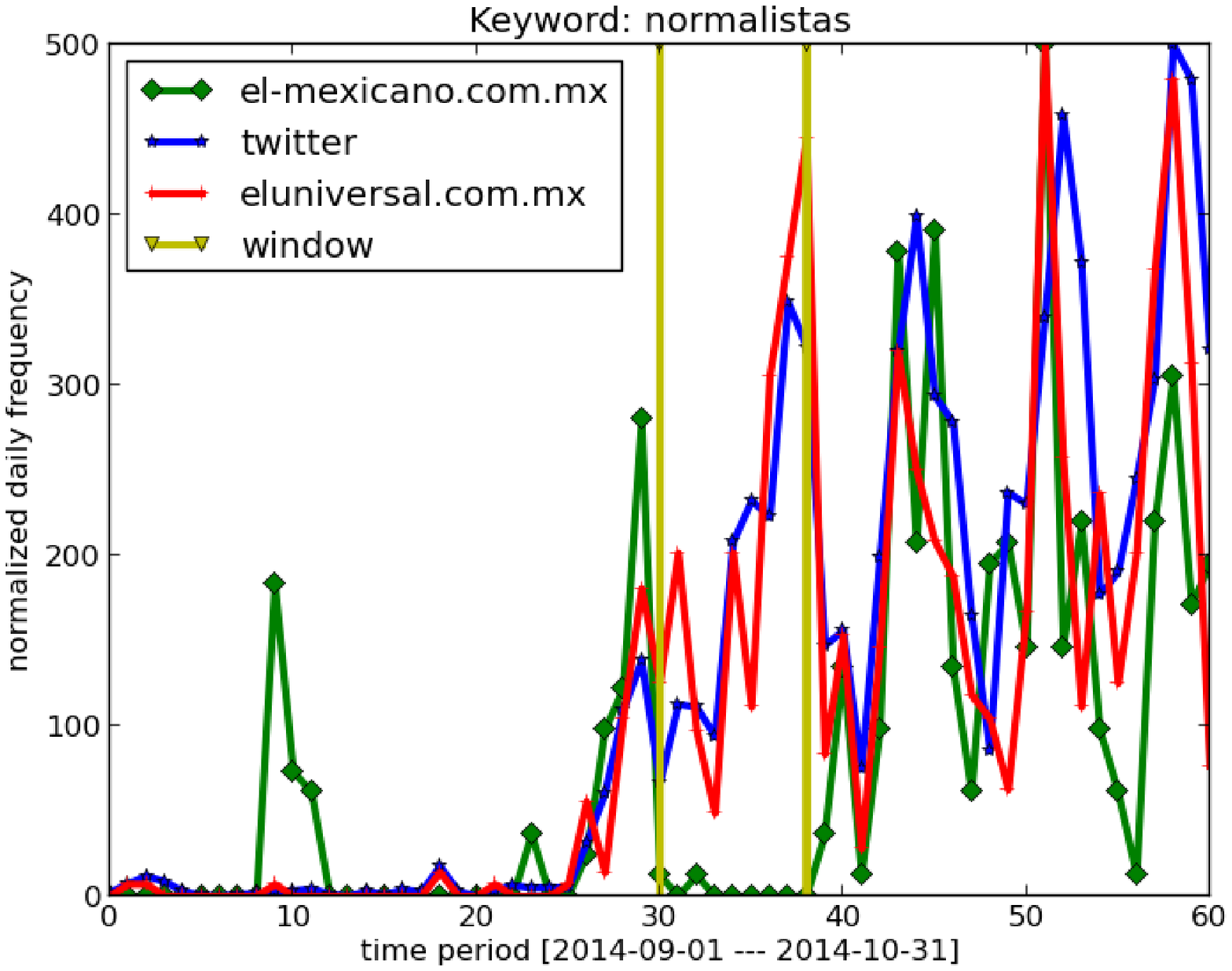}
    \caption{\scriptsize normalistas (students trained to become teachers)}
    \label{fig:2c}
    \end{subfigure}
    \begin{subfigure}[b]{0.22\textwidth}
            \centering
            \includegraphics[width=\textwidth]{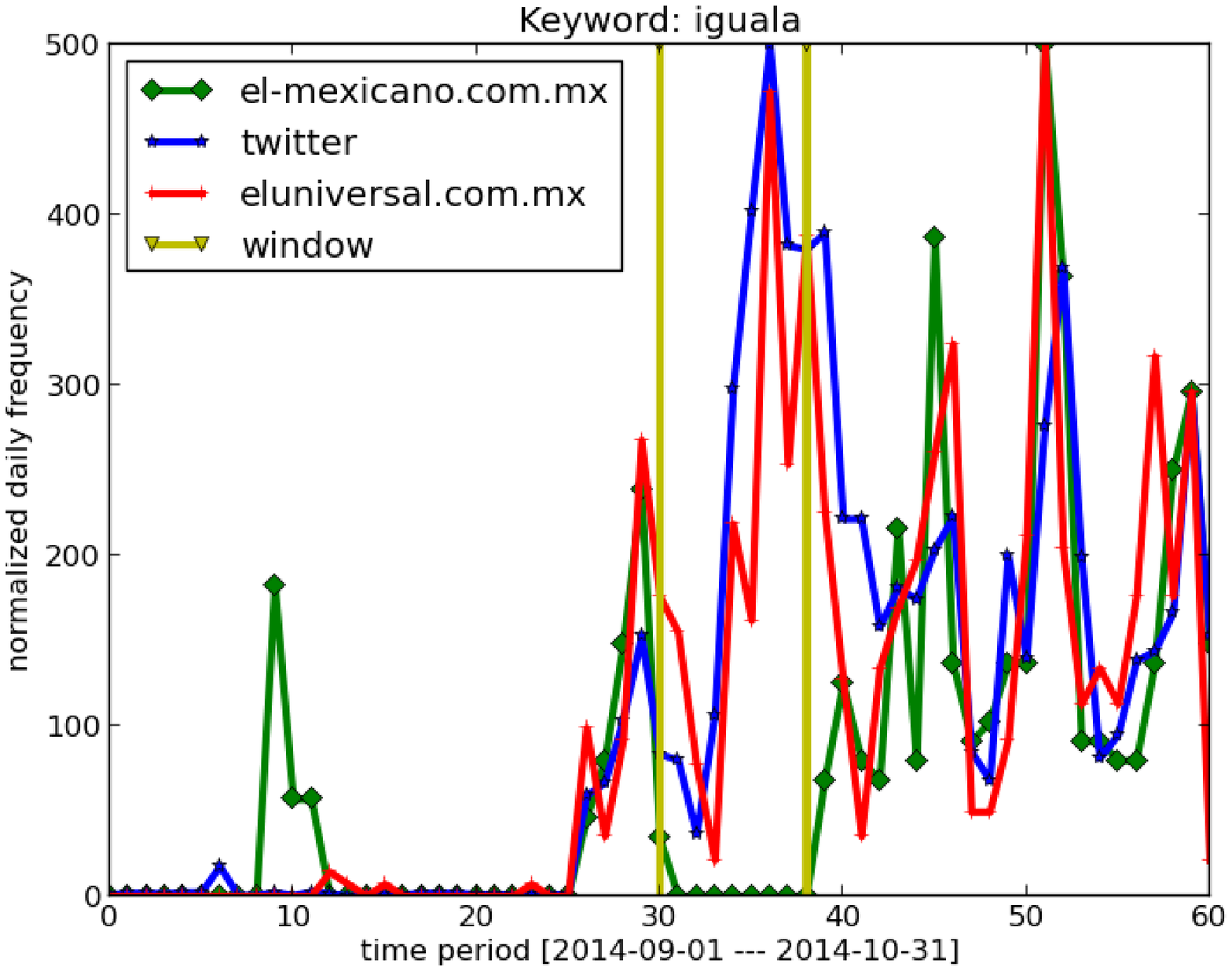}
    \caption{\scriptsize iguala}
    \label{fig:2d}
    \end{subfigure}
    \begin{subfigure}[b]{0.44\textwidth}
            \centering
            \includegraphics[width=\textwidth]{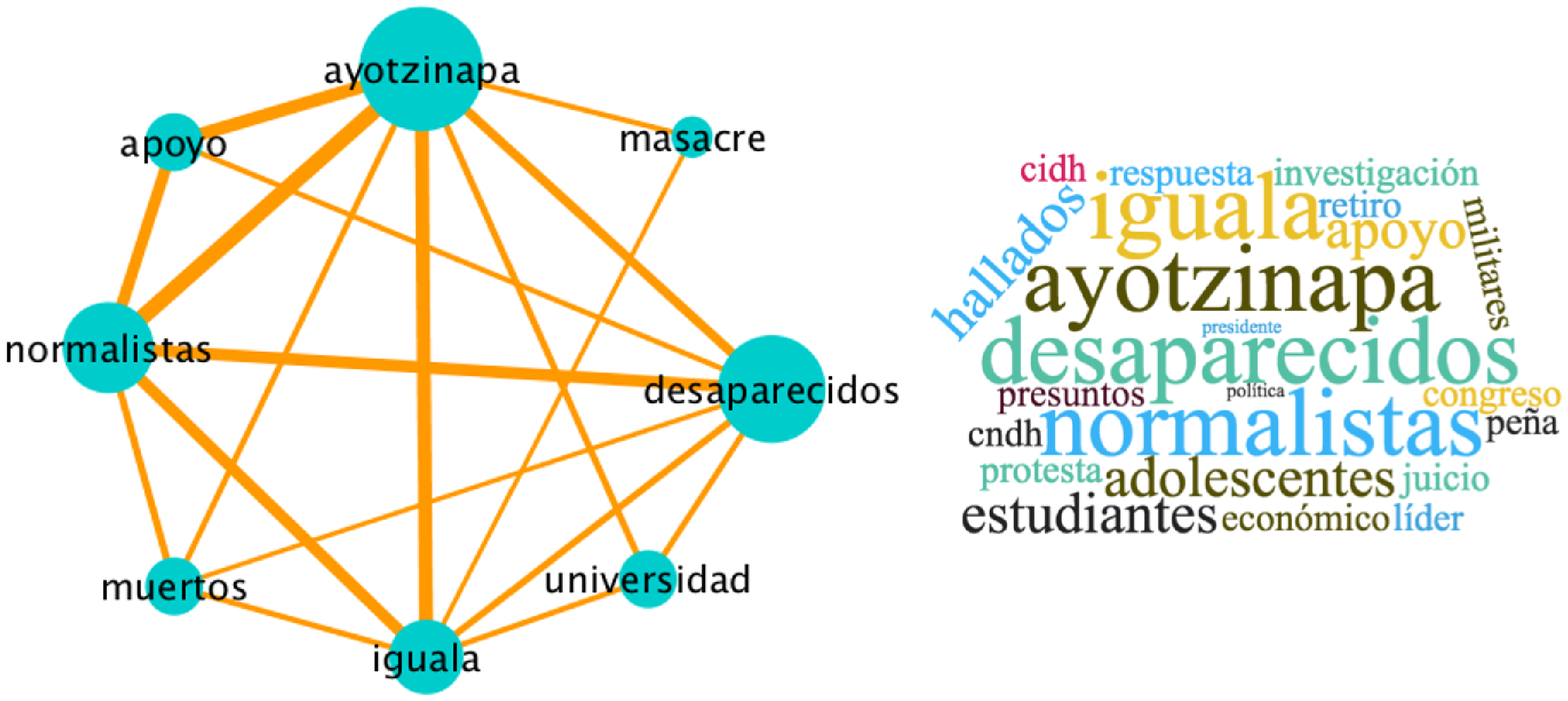}
    \captionsetup{font=small}
    \caption{\scriptsize Left: The strong connectivity of these keywords indicates their frequent co-occurrence in \textit{Twitter} and \textit{News}. A larger size of node indicates higher keyword frequency and a larger width of edge indicates more frequently co-occurrence; Right : word cloud representing censored keywords in \textit{News} around 2014-09-26 in Mexico}
    \label{fig:2e}
    \end{subfigure}
    \captionsetup{font=small}
    \caption{\scriptsize Example TSDF in \textit{News} vs. TSDF in \textit{Twitter} for a set of connected keywords. These keywords are relevant to the 43 missing students from Ayotzinapa Rural Teachers' College on Sep 26th, 2014 in Mexico. We can find consistent censorship pattern in El Mexicano Gran Diario Regional (el-mexicano.com.mx) shortly after the students are missing. 
    }
    \label{fig:introexample}
\end{figure}

\section{Methodology}

This section presents a novel hypothesis testing framework for characterizing the censorship patterns as discussed in Section~\ref{section:data analysis} and  an efficient inference algorithm for automatic detection of such censorship patterns in nearly linear time.
\subsection{Problem Formulation}
\label{subsection: problem formulation}
Suppose we have a dataset of news reports and a dataset of tweets within a shared time period in a country of interest. Each news report or tweet is represented by a set of keywords and is indexed by a time stamp (e.g., day). We model the joint information of news reports and tweets using an undirected keyword co-occurrence graph $\mathbb{G} = (\mathbb{V}, \mathbb{E})$, where $\mathbb{V} = \{1,2,\cdots$ $,n\}$ refers to the ground set of nodes/keywords, $n$ refers to the total number of nodes, and $\mathbb{E} \subseteq \mathbb{V} \times \mathbb{V}$ is a set of edges, in which an edge $(i, j)$ indicates that the keywords i and j co-occur in at least one news report or tweet. 
Each node $v \in \mathbb{V}$ is associated with four attributes: $\{a^t(v)\}_{t=1}^T$, $\lambda_a(v)$, $\{b^t(v)\}_{t=1}^T$, and $\lambda_b(v)$ as defined in Table~\ref{table:parameters}. As our study is based on the analysis of correlations between frequencies of keywords in the news and Twitter datasets, we only consider the keywords whose frequencies in these two datasets are well correlated (with correlations above a predefined threshold 0.15). Our goal is to detect a cluster (subset) of co-occurred keywords and a time window as an indicator of censorship pattern, such that the distribution of frequencies of these keywords in the news dataset is significantly different from that in the Twitter dataset.  

Suppose the chosen time granularity is day and the shared time period is $\{1, \cdots, T\}$. We consider two hypotheses: under the null ($H_0$),  the daily frequencies of each keyword $v$ in the news and Twitter datasets follow two different Poisson distributions with the mean parameters $\lambda_a(v)$ and $\lambda_b(v)$, respectively; under the alternative ($H_1(S,R)$), there is a connected cluster $S$ of  keywords 
and a continuous time window $R \subseteq \{1, \cdots, T\}$, in which the daily frequencies of each keyword $v$ in the Twitter dataset follow a Poisson with an elevated mean parameter $q_a \cdot \lambda_a(v)$, but those in the news dataset follows a Poisson with a down-scaled mean parameter $q_b \cdot \lambda_b(v)$. 
Formally, they can be defined as follows:
\begin{itemize}[leftmargin=*]
\item Null hypothesis $H_0$: 
\small
\begin{eqnarray} 
a^t(v) \sim \text{Pos}(\lambda_a(v)),
 \forall v \in \mathbb{V}, t \in \{1, \cdots, T\} \nonumber \\
 b^t(v) \sim \text{Pos}(\lambda_b(v)), \forall v \in \mathbb{V}, t \in \{1, \cdots, T\} \nonumber 
 \end{eqnarray}
 \normalsize
\item Alternative hypothesis $H_1(S, R)$:  
\small
\begin{align}
\noindent a^t(v) &\sim \text{Pos}(q_a \cdot \lambda_a(v)) \ , \  b^t(v) \sim \text{Pos}(q_b \cdot \lambda_b(v)), \forall v \in S , t \in R  \nonumber   \\ 
a^t(v) &\sim \text{Pos}(\lambda_a(v)), \ \  \ \ \ \ \  b^t(v) \sim \text{Pos}( \lambda_b(v)), \forall v \notin S \text{ or }
t \notin R \nonumber
\end{align}
\normalsize
\end{itemize}
where $q_a > 1, q_b < 1$, $S\subseteq \mathbb{V}$, the subgraph induced by $S$ (denoted as $\mathbb{G}_S$) must be connected to ensure that these keywords are semantically related, and $R \subseteq \{1,2,\cdots,T\}$ is a continuous time window defined as $\{i,i+1,\cdots,j\}, 1 \leq i \leq j \leq T$. Given the Poisson probability mass function denoted as $p(x;\lambda) = {\lambda^x e^{-\lambda}} / {x!}$, a generalized log likelihood ratio test (GLRT) statistic can then be defined to compare these two hypotheses, and has the form: 
\begin{small}
\begin{eqnarray}
F(S, R) = \log \frac{\max_{q_a > 1}\prod_{t \in R}\prod_{v \in S} p(a^t(v); q_a \lambda_a(v) ) }{\prod_{t \in R}\prod_{v \in S} p(a^t(v); \lambda_a(v) ) }\nonumber \\
+ \log \frac{\max_{q_b < 1}\prod_{t \in R}\prod_{v \in S} p(b^t(v); q_b \lambda_b(v) ) }{\prod_{t \in R}\prod_{v \in S} p(b^t(v); \lambda_b(v) ) }. 
\label{equation: F(S,R)}
\end{eqnarray}
\end{small}
In order to maximize the \textsc{GLRT} statistic, we need to obtain the maximum likelihood estimates of $q_a$ and $q_b$, which we set ${\partial F(S,R)} / {\partial q_a} = 0$ and ${\partial F(S,R)} / {\partial q_b} = 0$, respectively and get the best estimate $\hat{q_a} = C_a / B_a $ of $q_a$ and $\hat{q_b} = C_b / B_b $ of $q_b$ where $C_a = \sum_{v \in S, t \in R} a^t(v)$, $C_b =\sum_{v \in S, t \in R} b^t(v)$, $B_a = \sum_{v \in S, t \in R} \lambda_a(v)$, $B_b = \sum_{v \in S, t \in R} \lambda_b(v)$. Substituting $q_a$ and $q_b$ with the best estimations $\hat{q_a}$ and $\hat{q_b}$, we obtain the parametric form of the \textsc{GLRT} statistic as follows:
\begin{small}
\begin{equation}
F(S, R) = \Big( C_a \log \frac{C_a}{B_a} + B_a - C_a \Big) + \Big( C_b \log \frac{C_b}{B_b} + B_b - C_b \Big)
\label{equation: score function}
\end{equation}
\end{small}
Given the GLRT statistic $F(S,R)$, the problem of censorship detection can be reformulated as Problem 1 that is composed of two major components: 1) \textbf{Highest scoring clusters detection.} The highest scoring clusters are identified by maximizing the GLRT statistic $F(S,R)$ over all possible clusters of keywords and time windows;
2) \textbf{Statistical significance analysis.} The empirical p-values of the identified clusters are estimated via a randomization testing procedure~\citep{neill2009empirical}, and are returned as significant indicators of censorship patterns in the news dataset, if their p-values are below a predefined significance level (e.g., 0.05).

\begin{problem}{\textbf{(\textsc{GLRT} Optimization Problem)}}
\label{definition: problem1}
Given a keyword co-occurrence graph $\mathbb{G}(\mathbb{V},\mathbb{E})$ and a predefined significance level $\alpha$, the \textsc{GLRT} optimization problem is to find the set of highest scoring and significant clusters $\mathbb{O}$. Each  cluster in $\mathbb{O}$ is denoted as a specific pair of connected subset of keywords ($S_i \subseteq \mathbb{V}$) and continuous time window ($R_i \subseteq \{1, \cdots, T\}$), in which $S_i$ is the highest scoring subset within the time window $R_i$: 

\begin{small}
\begin{eqnarray}
\max\nolimits_{S \subseteq \mathbb{V}} F(S, R_i) {\text{ s.t. }} S \text{ is connected},
\label{equation: optimization LRT}
\end{eqnarray}
\end{small}

\noindent and is significant with respect to the  significance level $\alpha$. 
\end{problem}

\subsection{\textsc{GraphDPD} Algorithm}
\label{subsection: algorithm1}
Our proposed algorithm \textsc{GraphDPD} decomposes Problem~\ref{definition: problem1} into a set of  sub-problems, each of which has a fixed continuous time window, as 
shown in Algorithm~\ref{algorithm:Graph-DPD}. 
For each specific day $i$ (the first day of time window $R$ in Line 6) and each specific day $j$ (the last day of time window $R$ of Line 6), we solve the sub-problem (Line 7) with this specific $R= \{i,i+1,\cdots,j\}$ using $\textsc{Relaxed-GrapMP}$ algorithm which will be elaborated later. For each connected subset of keywords $S$ returned by $\textsc{Relaxed-GraphMP}$, its p-value is estimated by randomization test procedure~\cite{neill2009empirical}(Line 8). The pair $(S, R)$ will be added into the result set $\mathbb{O}$ (Line 9) if its empirical p-value is less than a predefined significance level $\alpha$ (e.g., 0.05). The procedure getPValue in Line 8 refers to a randomization testing procedure based on the input graph $\mathbb{G}$ to calculate the empirical p-value of the pair $(S, R)$~\citep{neill2009empirical}. Finally, we return the set $\mathbb{O}$ of signifiant clusters as indicators of censorship events in the news data set.
\begin{algorithm}[ht]
\caption{\textsc{GraphDPD}}
\begin{algorithmic}[1]
 \State \textbf{Input}: Graph Instance $\mathbb{G}$ and significant level $\alpha$;
 \State \textbf{Output}: set of anomalous connected subgraphs $\mathbb{O}$;
 \State $\mathbb{O} \leftarrow \emptyset$;
 \For{ $i \in \{1,\cdots,T\}$}
 \For{ $j \in \{i+1,\cdots,T\}$}
 \State $R \leftarrow \{i,i+1,\cdots,j\}$ ;\ {//} \textit{time window $R$} 
 \State $S \leftarrow \textsc{Relaxed-GraphMP}(\mathbb{G},R)$;
\If{getPValue($\mathbb{G}$, $S$, $R$) $ \leq \alpha$}
\State $\mathbb{O} \leftarrow \mathbb{O} \cup (S, R)$;
\EndIf
\EndFor
\EndFor
\State \textbf{return} $\mathbb{O}$;    \end{algorithmic}
\label{algorithm:Graph-DPD}
\end{algorithm}
Line 7 in Algorithm~\ref{algorithm:Graph-DPD} aims to solve an instance of Problem~\ref{definition: problem1} given a specific time window $R$, which is a set optimization problem subject to a connectivity constraint. Tung-Wei et. al. \cite{kuo2015maximizing} proposed an approach for maximizing submodular set function subject to a connectivity constraint on graphs. However, our objective function $F(S,R)$ is non-submodular as shown in Theorem~\ref{nonsubmodular-proof} and this approach is not applicable here.
\begin{theorem}
\label{nonsubmodular-proof}
Given a specific window $R$, our objective function $F(S,R)$ defined in~(\ref{equation: score function}) is non-submodular.
\end{theorem}

We propose a novel algorithm named \textsc{Relaxed-GraphMP} to approximately solve  Problem \ref{definition: problem1} in nearly linear time with respect to the total number of nodes in the graph. 
We first transform the GLRT statistic in Equation(\ref{equation: score function}) to a vector  form. Let $\bf x$ be an n-dimensional vector $(x_1,x_2,\cdots,x_n)^{\mathsf{T}}$, where $x_i \in \{0,1\}$ and $x_i = 1$ if $i \in S$, $x_i = 0$ otherwise. We define  $\mathcal{P},\mathcal{Q},\Lambda_a, \Lambda_b$ as follows:

\begin{small}
\begin{align}
\mathcal{P} = \Bigg[\sum_{t\in R} a^t(1),\cdots,\sum_{t\in R} a^t(n) \Bigg]^{\mathsf{T}} , \Lambda_a &= [\lambda_{a}(1),\cdots,\lambda_{a}(n)]^{\mathsf{T}}, \nonumber \\
\mathcal{Q} = \Bigg[\sum_{t\in R} b^t(1),\cdots,\sum_{t\in R} b^t(n) \Bigg]^{\mathsf{T}}, \Lambda_b &= [\lambda_{b}(1),\cdots,\lambda_{b}(n)]^{\mathsf{T}}. \nonumber
\end{align}
\end{small}

\noindent Therefore, $C_a$, $C_b$, $B_a$, and $B_b$ in Equation(\ref{equation: score function}) can be reformulated as follows: 
\begin{small}
\begin{align}
C_a &= \mathcal{P}^{\mathsf{T}}{\bf x}, \ \ C_b = \mathcal{Q}^{\mathsf{T}} {\bf x}, \ \ B_a = |R| {\Lambda_a}^{\mathsf{T}} {\bf x}, \ \ B_b = |R| {\Lambda_b}^{\mathsf{T}} {\bf x} \nonumber
\end{align}
\end{small}
Hence, $F$ can be reformulated as a relaxed function $\hat{F}$: 
\begin{small}
\begin{flalign}
\label{equation: relaxedF}
\hat{F}({\bf x},R) &= \mathcal{P}^{\mathsf{T}}{\bf x} \log \frac{\mathcal{P}^{\mathsf{T}}{\bf x}}{|R|{\Lambda_a}^\mathsf{T} {\bf x}} + {|R|{\Lambda_a}^\mathsf{T} {\bf x} - \mathcal{P}^{\mathsf{T}} {\bf x}} \nonumber \\
& + \mathcal{Q}^{\mathsf{T}}{\bf x} \log \frac{\mathcal{Q}^{\mathsf{T}}{\bf x}}{|R|{\Lambda_b}^\mathsf{T} {\bf x}} + {|R|{\Lambda_b}^\mathsf{T} {\bf x} - \mathcal{Q}^{\mathsf{T}} {\bf x}}
\end{flalign}
\end{small}

\noindent 
We relax the discrete domain $\{0,1\}^n$ of $S$ to the  continuous domain $[0,1]^n$ of ${\bf x}$, and obtain  the relaxed version of Problem~\ref{definition: problem1} as described in Problem~\ref{problem: problem2}.  

\begin{problem}{\textbf{Relaxed \textsc{GLRT}
\label{problem: problem2}
Optimization Problem}} Let $\hat{F}$ be a continuous surrogate function of $F$ that is defined on the relaxed domain $[0,1]^n$ and is identical to $F(S,R)$ on the discrete domain $\{0,1\}^n$. The relaxed form of \textbf{\textsc{GLRT} Optimization Problem} is defined the same as the \textsc{GLRT} optimization problem, except that, for each pair $(S_i, R_i)$ in $\mathbb{O}$, the subset of keywords $S_i$ is identified by solving the following problem with $S_i = \text{supp}(\hat{\bf x})$: 
\begin{eqnarray}
    \hat{\bf x} = \arg\max_{{\bf x} \in [0,1]^n} \hat{F}({\bf x}, R_i)\ \  s.t.\ \  \text{supp}({\bf x}) \  \text{is connected.}\nonumber 
\end{eqnarray}
where $\text{supp}({\bf x}) = \{i | x_i \ne 0\}$ is the support of ${\bf x}$. The gradient of ${\hat{F}}({\bf x},R)$ has the form: 
\begin{small}
\begin{flalign}
    \frac{\partial \hat{F}({\bf x},R)}{\partial {\bf x}}
    &= \log \frac{\mathcal{P}^{\mathsf{T}}{\bf x}}{|R|{\Lambda_a}^{\mathsf{T}}{\bf x}}\mathcal{P} + \Big( |R| - \frac{\mathcal{P}^{\mathsf{T}} {\bf x}}{{\Lambda_a}^{\mathsf{T}}{\bf x}}\Big){\Lambda_a} \nonumber \\
    &+ \log \frac{\mathcal{Q}^{\mathsf{T}}{\bf x}}{|R|{\Lambda_b}^{\mathsf{T}}{\bf x}}\mathcal{Q} + \Big( |R| - \frac{\mathcal{Q}^{\mathsf{T}} {\bf x}}{{\Lambda_b}^{\mathsf{T}}{\bf x}}\Big)\Lambda_b
    \label{equation: gradient}
\end{flalign}
\end{small}
\end{problem}

\begin{algorithm}
\caption{\textsc{Relaxed-GraphMP}}
\begin{algorithmic}[1]
 \State \textbf{Input}: Graph instance $\mathbb{G}$, continous time window $R$;
 \State \textbf{Output}: the co-occurrence subgraph $\mathbb{G}_S$;
  \State $i \leftarrow 0$; ${\bf x}^i \leftarrow \text{an initial vector}$;
  \Repeat 
  \State $\nabla \hat{F}({\bf x}^i,R) \leftarrow \frac{\partial {\hat{F}}({{\bf x}^i},R)}{\partial{\bf x}^i}\text{ by Equation~(\ref{equation: gradient})}$ ;
  \State ${\bf g} \leftarrow \textbf{Head}(\nabla \hat{F}({\bf x}^i,R),\mathbb{G})$;\ {//} \textit{Head projection step} 
  \State $\Omega \leftarrow \text{supp}({\bf g})\cup \text{supp}({\bf x}^i)$;
  \State ${\bf b} \leftarrow \arg \max_{{\bf x} \in [0, 1]^n} \hat{F}({\bf x},R) \ \ s.t. \ \  \text{supp}({\bf x}) \subseteq \Omega$; 
  \State ${\bf x}^{i+1} \leftarrow \textbf{Tail}({\bf b},\mathbb{G})$;\ {//} \textit{Tail projection step}
  \State $i\leftarrow i+1, S \leftarrow \text{supp}({\bf x}^i)$;
  \Until{halting condition holds;}
  \State \textbf{return} $(S, R)$;
  \end{algorithmic}
  \label{algorithm:GraphMP}
\end{algorithm}

 Our proposed algorithm \textsc{Relaxed-GraphMP} decomposes Problem~\ref{problem: problem2} into two sub-problems that are easier to solve: 1) a single utility maximization problem that is independent of the connectivity constraint; and 2) head projection and tail projection problems \cite{hegde2015nearly} subject to connectivity constraints. We call our method \textsc{Relaxed-GraphMP} which is analogous to \textsc{GraphMP} proposed by Chen et al.~\cite{chengeneralized}.
The high level of \textsc{Relaxed-GraphMP} is shown in Algorithm~\ref{algorithm:GraphMP}. It contains 4 main steps as described below.
\begin{itemize}[leftmargin=*]
\item \textbf{Step 1:} Compute the gradient of relaxed \textsc{GLRT} problem (Line 5). The calculated gradient is $\nabla \hat{F}({\bf x}^i,R)$. Intuitively, it maximizes this gradient with connectivity constraint that will be solved in next step.
\item \textbf{Step 2:} Compute the head projection (Line 6). This step is to find a vector ${\bf g}$ so that the corresponding subset $\text{supp}({\bf g})$ can maximize the norm of the projection of gradient $\nabla \hat{F}({\bf x}^i,R)$ ( See details  in \cite{hegde2015nearly}).
\item \textbf{Step 3:} Solve the maximization problem without connectivity constraint. This step (Line 7,8) solves the maximization problem subject to the $\text{supp}({\bf x}) \subseteq \Omega$, where $\Omega$ is the union of the support of the previous solution $\text{supp}({\bf x}^i)$ with the result of head projection $\text{supp}({\bf g})$ (Line 7). A gradient ascent based method is proposed to solve this problem. Details is not shown here due to space limit. 
\item \textbf{Step 4:} Compute the tail projection (Line 9). This final step is to find a subgraph $\mathbb{G}_S$ so that ${\bf b}_S$ is close to ${\bf b}$ but with connectivity constraint. This tail projection guarantees to find a subgraph $\mathbb{G}_S$ with constant approximation guarantee (See details in \cite{hegde2015nearly}). 
\item \textbf{Halting:}  The algorithm terminates when the condition holds. Our algorithm returns a connected subgraph $\mathbb{G}_S$ where the connectivity of $\mathbb{G}_S$ is guaranteed by Step 4.
\end{itemize}


\noindent \textbf{Time Complexity Analysis:} The \textsc{GraphDPD} algorithm is efficient as its time complexity is proportional to the total number of continous time windows $T^2$. Therefore, the time complexity of \textsc{GraphDPD} is mainly dependent on the run time of \textsc{Relaxed-GraphMP}. We give the detailed time complexity analysis in Theroem~\ref{theorem: timeComplexity}.
\begin{theorem} \textsc{GraphDPD} runs in $O(T^2\cdot t( n T + n l + |\mathbb{E}|{\log^3 n})) $ time, where $T$ is the maximal time window size, $nT$ is the time complexity of Line 5 in Algorithm~\ref{algorithm:Graph-DPD}, $nl$ is the run time of Line 8 using gradient ascent, $|\mathbb{E}|{\log^3 n}$ is the total run time of head projection and tail projection algorithms, and $t$ is the total number of iterations needed in Algorithm~\ref{algorithm:GraphMP}.
\label{theorem: timeComplexity}
\end{theorem}
\begin{proof}
As the maximal time window in input graph $\mathbb{G}$ is $T$, \textsc{GraphDPD} needs $O(T^2)$ iterations in its inner loop and outer loop (From Line 4 to Line 11 in Algorithm~\ref{algorithm:Graph-DPD}) to execute \textsc{Relaxed-GraphMP} (Line 7). Suppose \textsc{Relaxed-GraphMP} needs $t$ iterations, the time complexity of each iteration has three parts: 1). $O(nT)$ the run time for calculating gradient in Line 5 of Algorithm~\ref{algorithm:GraphMP}; 2). $O(nl)$: the run time of Line 8 using gradient ascent where $l$ is the number of iterations in gradient ascent method; and 3). $O(|\mathbb{E}|\log^3 n)$: the run time of head and tail projection in Line 6 and Line 9 of Algorithm~\ref{algorithm:GraphMP}. Hence the time complexity of \textsc{Relaxed-GraphMP} is $t(n T + n l + |\mathbb{E}|\log ^3 n)$. Therefore, the total time complexity of \textsc{GraphDPD} immediately follows. As observed in our experiments, the numbers of iterations, including $t$ and $l$, scale constant with respect to $n$, and the overall time complexity of \textsc{GraphDPD} hence scales nearly linear with respect to $n$. 
\end{proof}

\section{Experiments}
Through experiments, we (1) evaluate the performance of our proposed approach in censorship pattern detection compared with baseline methods, and (2) perform case studies that analyze the censorship patterns we have found in real data. The code and datasets will be available to researchers for evaluation purposes.
\subsection{Experimental Design}
\label{section: experimental_design}

\textbf{Real world datasets:} Table \ref{table:realdata} gives a detailed description of real-world datasets we used in this work. Details of Twitter and news data access have been provided in Section 3. Daily keyword co-occurrence graphs, which integrate \textit{News} with \textit{Twitter}, are generated as described in Section~\ref{section:data preprocessing}. \\ \indent

\begin{small}
\begin{table}[h!]
\scriptsize
\centering
\caption{\scriptsize Real-world dataset used in this work. Tweets: average number of daily tweets. News: average number of daily local news articles. Number of nodes and edges are averaged over daily keyword co-occurrence graphs.}
\begin{tabular}{|c c c c c|} 
 \hline
 Country & Daily Tweets & Daily News & \# of Nodes & \# of Edges \\ [0.6ex] 
 \hline
 Mexico & 249,124 & 863 & 3,369 & 93,919 \\ 
 Venezuela & 222,072 & 169 & 2,469 & 37,740 \\[1ex] 
 \hline
\end{tabular}
\label{table:realdata}
\end{table}
\end{small}

\textbf{Data Preprocessing: } The preprocessing of the real world datasets has been discussed in detail in Section~\ref{section:data preprocessing}. In particular, we considered keywords whose day-by-day frequencies in news media and Twitter data have linear correlations above 0.15, in order to filter noisy keywords. 

\textbf{Semi-synthetic datasets:} We create semi-synthetic datasets by using the coordinates from real-world datasets and injecting anomalies. 

Ten daily keyword co-occurrence graphs are randomly selected to inject with random true anomaly connected subgraphs using a random walk algorithm~\cite{tong2006fast} with a restart probability of $0.1$. The number of nodes in the true anomaly subgraph is $x$ percentage of the number of nodes in the daily co-occurrence graph, where $x \in \{0.05, 0.1, 0.15\}$. For convenience but without loss of generality, we specified $q_t \cdot q_n=1.0$, where $q_t$ controls the scale of anomaly in tweets and $q_n$ controls scale of anomaly in local news articles. In our experiments, we set $q_t = \{1.0,2.0,\cdots,10.0,15.0,$ $\cdots,35.0\}$, and set 
$q_n = 1 / q_t$ correspondingly. \\
\indent
\textbf{Our proposed \textsc{Graph-DPD} and baseline methods: }The maximal window size $T$ and result threshold $\alpha$ in \textsc{Graph-DPD} are set as 7 and 0.05 respectively. However, our algorithm is not sensitive to the setting of $T$ and $\alpha$. We compare our proposed method with one baseline method LTSS \citep{neill2012fast}, which finds anomalous but not necessarily connected subsets of data records by maximizing a score function.
We also compare our proposed method with two state-of-art baseline methods designed specifically for connected anomalous subgraph detection, namely, EventTree \citep{rozenshtein2014event} and NPHGS \citep{chen2014non}. Model parameters are tuned following the original papers.
Specifically, for EventTree we tested $\lambda=\{0.0001$, $0.0006$, $\cdots$ ,$0.001$, $0.006$, $\cdots$, $0.010$, $0.015$, $\cdots$,$ $ $0.1$,$0.5$, $1.0$,$\cdots$ ,$20.0\}$. For NPHGS, we set the number of seed entities  $K = 400$ and  typical significance levels $\alpha_{max} = 0.15$ as the authors suggested. Since the baseline methods are designed to detect anomalies on one data source at one time, they are tested separately on \textit{Twitter} and \textit{News}, which are labeled as
$\text{LTSS}_\text{News}$, $\text{LTSS}_\text{Twitter}$, $\text{EventTree}_\text{News}$, $\text{EventTree}_\text{Twitter}$, 
$\text{NPHGS}_\text{News}$ and $\text{NPHGS}_\text{Twitter}$. Sepecifically, $\text{LTSS}_\text{Twitter}$, \\ $\text{EventTree}_\text{Twitter}$ and $\text{NPHGS}_\text{Twitter}$ are burst detection baseline methods while $\text{LTSS}_\text{News}$, $\text{EventTree}_\text{News}$, and $\text{NPHGS}_\text{News}$ are absenteeism detection baseline methods by some transformations on attributes.

\textbf{Performance Metrics:}
The performance metrics include: (1) precision (Pre), (2) recall (Rec), and (3) f-measure (F-score). Given the returned subset of nodes $S$ and the corresponding true subset of anomalies $S^*$, we can calculate these metrics as follows:
\begin{small}
\begin{equation}
 \text{Pre} = \frac{|S \cap S^*|}{ |S|}, \text{Rec} = \frac{|S \cap S^*|}{ |S^*|},\text{F-score} = \frac{2|S \cap S^*|}{ |S^*| + |S|} \nonumber   
\end{equation}
\end{small}
\indent
\textbf{Collecting labels for real data:} We collect labels for real-world instances of censorship from all abnormal absence patterns identified in \textit{News} by all baseline methods. For each abnormal absence pattern in \textit{News}, we need to first identify if there are any relevant events of interest taking place around the associated time region. An indicator of censorship pattern is considered as valid if we can find the event of interest is: 1) not reported in some local news outlets while reported in some different local news outlets, 2) reported in influential international news outlets, and 3) reported of censorship activity in local news media from other news outlets during the associated time window. The evaluation process is analyzed with the inner-annotator agreement by multiple independent annotators (5 of the authors of this paper).
\begin{figure*}[h!]
\centering
\includegraphics[scale=0.5]{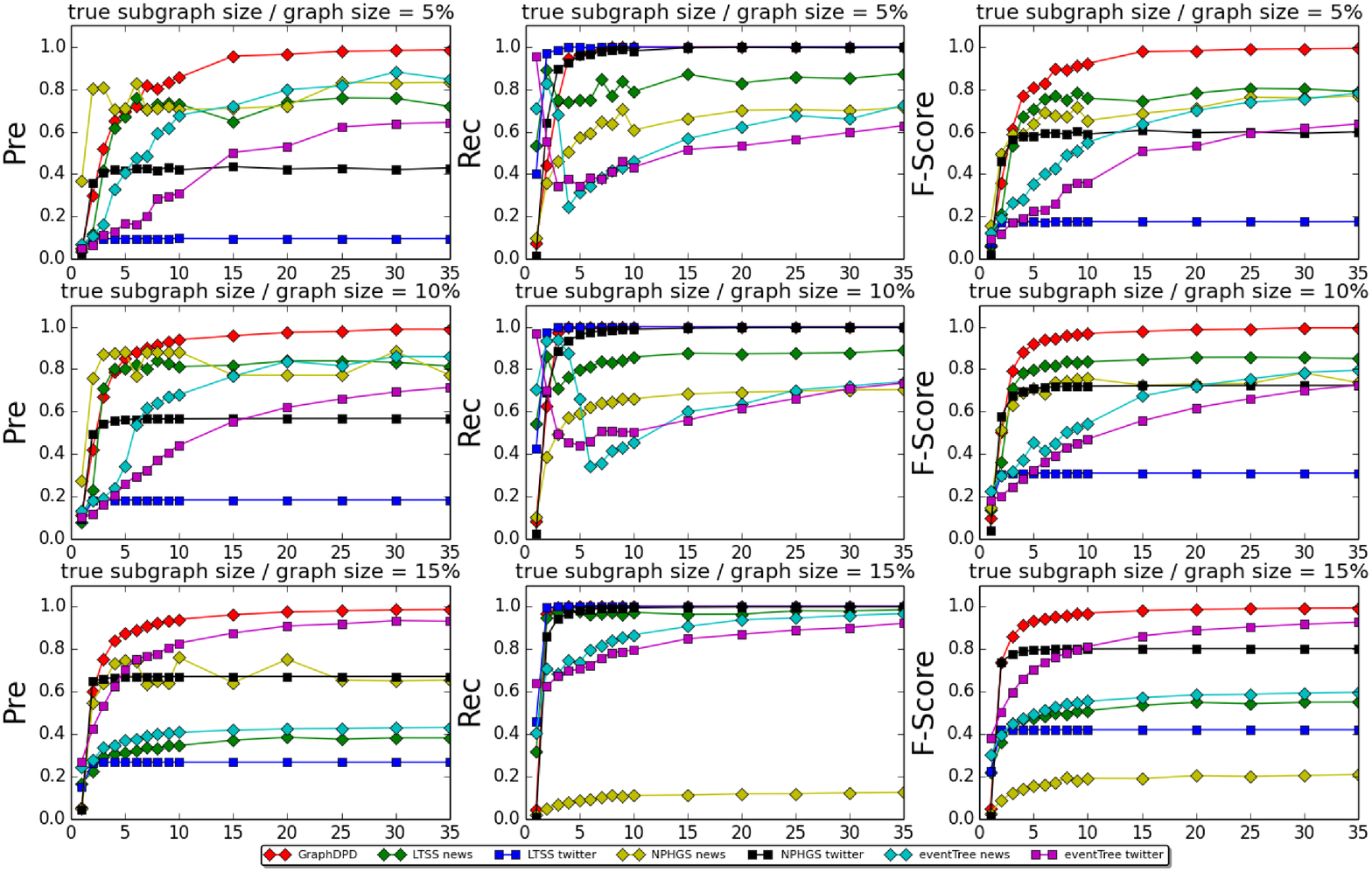}
\caption{\scriptsize Anomaly detection results (mean precision (left), recall (center), and F-measure (right) vs. perturbation intensity) for different anomaly subgraph sizes (increased size from top to bottom) in semi-synthetic data. X-axis represents $q_t$, which implies the scale of anomaly injected in \textit{Twitter}. $q_n$, which implies the scale of anomaly injected in \textit{News}, is varied following $q_t \times q_n=1.0$.}
\label{fig:simulation}
\end{figure*}
\subsection{Semi-synthetic Data Evaluation}
\label{subsection:5.2}
 \indent
We evaluate the accuracy of our approach to detect the disrupted ground truth anomalies. Fig. \ref{fig:simulation} shows the average precision, recall, and F-measure in detecting the injected anomalies using the semi-synthetic data. We find that overall our approach consistently outperforms all other baseline methods. \\ \indent
\textbf{Detection power.} (1) \textbf{Our approach.} Our approach outperforms baseline methods especially at low perturbation intensities where the detection is harder to carry out, and the performance increases gradually with the increase of perturbation intensity. In particular, it has a high accuracy of detecting injected anomalies when $q_t \geq 10$ regardless of the size of injected anomalies. Measures of recall using $\text{NPHGS}_\text{Twitter}$ are as good as our approach while the other baseline methods are significantly worse than our approach especially when the size of disrupted ground truth anomalies is small and perturbation intensity is low. However, the measures of precision using $\text{NPHGS}_\text{Twitter}$ are much worse than our approach. Considering overall F-score, $\text{NPHGS}_\text{News}$ and $\text{NPHGS}_\text{Twitter}$ look similar to our approach when perturbation intensity is low while much worse than our approach when perturbation intensity is high. When we increase $q_t$, EventTree based methods perform worse than our approach, especially when the size of ground truth anomalies is small. (2) \textbf{$\text{NPHGS}$.} When $q_t \in \{1.0, 2.0 \}$ and true ratio $x \in \{0.05,0.10\}$, the precision of $\text{NPHGS}_\text{News}$ is better than our method. However, when $x=0.15$, the recall of $\text{NPHGS}_\text{News}$ becomes quite low, which indicates its poor behavior when true subgraph is relatively large. (3) \textbf{$\text{EventTree}$.} The recall of $\text{EventTree}_\text{News}$ and $\text{EventTree}_\text{Twitter}$ is among the best when $q_t$ is less than 2.0. The reason is that results of EventTree are easier affected by noise nodes. (4) \textbf{$\text{LTSS}$.} In general, LTSS did well in average recall but poorly in average precision as the size of anomalous graph increases. Hence, our approach outperforms the baseline by detecting connected clusters of keywords.

\subsection{Real Data Evaluation}
\label{section:realdata}
We apply our proposed approach to \textit{Twitter} and \textit{News} of Mexico and Venezuela during Year 2014 as shown in Table \ref{table:realdata}. Performance evaluation on censorship detection is two-fold: (1) quantitative evaluation on accuracy, and (2) qualitative case studies.
\subsubsection{Quantitative Evaluation}
In this paper, we focus on a list of nine local news outlets for each country and test on all possible continuous time windows from 3 to 7 days starting from January 1, 2014 to December 25, 2014. For every local news outlet and time window, our approach finds the connected cluster of keywords that maximizes the objective function as defined in Eqn. \ref{equation: score function}. Consider the existence of baseline level of variation, we perform 5,000 random permutations and record the function scores associated with each randomly selected connected cluster of keywords. For every local news outlet, we remove connected clusters of keywords whose p-values are greater than a predefined significance level (0.05). The number of remaining connected cluster of keywords for each local news outlet is summarized in the second column of Table \ref{table:realdata-stat}. In order to eliminate overlapping time regions, connected clusters of keywords are ranked based on their p-values from low to high and merged if within 5 days of another connected cluster of keywords with a lower p-value. The number of distinct connected cluster of keywords for each local news outlet is summarized in the third column of Table \ref{table:realdata-stat}. For simplicity, we call each distinct connected cluster of keywords and its corresponding time window as an indicator of censorship pattern.  \\ \indent
For each country of interest, we group indicators of censorship patterns across all news outlets and merge similar indicators of censorship pattern in different news outlets. Two similar indicators of censorship pattern need to have overlapping time windows and overlapping event-relevant keywords. This yields 23 distinct indicators of censorship patterns detected in Mexico during Year 2014 and 5 of them are detected in all of the nine local news outlets. As discussed in Section \ref{section:pattern-analysis}, absence of patterns in all news outlets could be due to topical differences
between social and news media, and thus we use
the remaining 18 indicators of censorship for
evaluation. Similarly, we mined a total of 17 distinct indicators of censorship patterns in Venezuela during Year 2014 and 14 of them are considered in this study. \\ \indent

As discussed previously, existing approaches on censorship detection in social media rely on the collection of deleted posts and such approaches are not able to detect self censorship in news media. Hence, we apply three anomaly detection baseline methods, LTSS, NPHGS, and EventTree, to detect anomalies in \textit{News} on graphs with all possible time windows from 3 days to 7 days using starting days from January 1, 2014 to December 25, 2014. The parameters used for the baselines are set similarly as in Section \ref{section: experimental_design}. The baseline methods can find anomalous subgraphs according to their own score functions; however, they are not able to evaluate the significance level of each subgraph. For the purpose of comparison, we remove duplicate subgraphs with overlapping time regions in the same manner as our method. The remaining subgraphs are ranked from the best to the worst according to their function values and top 18 subgraphs in Mexico and top 14 subgraphs in Venezuela are selected to compare with our method. \\ \indent

Table \ref{table:evaluation} summarizes the comparison of false positive rates in censorship detection and our method outperforms LTSS, NPHGS, and EventTree. The baseline methods, which are designed for event detection instead of censorship detection, will capture all falling patterns in \textit{News}. 
In particular, the baseline methods are not able to successfully differentiate censored events from non-censored events, e.g., the normal end of attention paid to breaking events. 
Table \ref{table:top5} summarizes a list of example instances of censorship identified by our approach in Mexico and Venezuela with significance level $\leq 0.05$. We will next evaluate these instances.
\begin{small}
\begin{table*}
\scriptsize
\centering
\caption{\scriptsize A list of local news media used in our work and number of connected clusters of keywords detected by our approach during Year 2014}
\label{table:realdata-stat}
\begin{tabular}{|l|c|c|c|c|}
\hline
\multicolumn{3}{|c|}{{\bf Mexico}}            \\ \hline
News Media    & \pbox{2cm}{\# of results with p-value $\leq 0.05$}   & \# of distinct results   \\ \hline
El Imparcial & 56 &  10  \\ \hline
El Mexicano Gran Diario Regional    & 77   &  15    \\ \hline
El Siglo de Torreon        & 43 &  9 \\ \hline
El Universal in Mexico City  &  52  & 11 \\ \hline
El Informador   & 59  & 11  \\ \hline
Noroeste  &  65  & 13  \\ \hline
Novedades Acapulco  &  63  & 12  \\ \hline
Correo  &  58 & 11  \\ \hline
Vanguardia   & 58 & 11  \\ \hline
\multicolumn{3}{|c|}{{\bf Venezuela}}               \\ \hline
News Media    & \pbox{2cm}{\# of results with p-value $\leq 0.05$}   & \# of distinct results   \\ \hline
El Tiempo in Trujillo  & 59   & 11  \\ \hline
El Impulso & 58   & 9  \\ \hline
El Mundo    & 45   & 8  \\ \hline
El Nacional        & 57   & 11   \\ \hline
El Tiempo in Anzoategui   & 63   & 12  \\ \hline
El Universal in Caracas   & 46   & 7  \\ \hline
La Verdad   &   69   &  12  \\ \hline
Tal Cual  & 68   & 12  \\ \hline
Ultimas Notícias   & 76   &  14  \\ \hline
\end{tabular}
\end{table*}
\end{small}
\begin{small}
\begin{table}[h!]
\scriptsize
\centering
\caption{\scriptsize Comparison of false positive rates in censorship detection between \textsc{GraphDPD} and three baseline methods: LTSS, NPHGS, and EventTree on real data of Mexico and Venezuela during Year 2014.}
\begin{tabular}{|c c c c c|} 
 \hline
 Country & LTSS & NPHGS & EventTree & \textsc{GraphDPD} \\ [0.6ex] 
 \hline
 Mexico & 0.722  &  0.667 & 0.556 & 0.278 \\ 
 Venezuela & 0.714  &   0.786 & 0.643 & 0.357 \\[1ex] 
 \hline
\end{tabular}
\label{table:evaluation}
\end{table}
\end{small}

\begin{small}
\begin{table*}
\tiny
\centering
\caption{\scriptsize Example indicators of censorship identified by our approach in Mexico and Venezuela during Year 2014 (with significance level $\leq 0.05$)}
\vspace{-2mm}
\label{table:top5}
\begin{tabular}{| c | c | c | c |}
\hline
\multicolumn{4}{|c|}{Mexico} \\ \hline
Date   & Example censored keywords  &  \pbox{3cm}{Example local news media \\ detected with censorship} &  Reasons for censorship in news media     \\ \hline
\textcolor{orange}{2014-05-01}  & \pbox{5cm}{\textcolor{red}{reforma(reform)},  \textcolor{green}{gasolina(petrol)},  \textcolor{blue}{educación(education)}}     &  Noroeste   & \pbox{7cm}{Tens of thousands of people marched in Mexico City on \textcolor{orange}{Labor Day} to protest the new laws, which target at Mexico's \textcolor{blue}{education} system and \textcolor{red}{opening up} the state controlled \textcolor{green}{oil} industry to foreign investors.} \\ \hline
2014-09-27  & \pbox{5cm}{\textcolor{orange}{ayotzinapa}, \textcolor{magenta}{iguala}, \textcolor{red}{normalistas}, \textcolor{blue}{desaparecidos(missing)},  \textcolor{green}{detenidos(detained)}, protesta(protest)}   &  El Mexicano Gran Diario Regional   & \pbox{7cm}{43 \textcolor{red}{students} from the \textcolor{orange}{Ayotzinapa} Rural Teachers' College went \textcolor{blue}{missing} and \textcolor{green}{kidnapped} in \textcolor{magenta}{Iguala} on September 26, 2014. This incident became the biggest political and public security scandal  of Mexican President.}  \\ \hline
2014-11-10  & \pbox{5cm}{\textcolor{orange}{ayotzinapa}, \textcolor{red}{estudiantes(students)}, \textcolor{red}{normalistas}, \textcolor{blue}{desaparecidos(missing)}, \textcolor{green}{protesta(protest)}, militares(military), \textcolor{magenta}{iguala}}  &   El Mexicano Gran Diario Regional    & \pbox{7cm}{\textcolor{green}{Protests} in Mexico City demanding the return of the \textcolor{blue}{missing} \textcolor{red}{students}, who came from \textcolor{orange}{Ayotzinapa} Rural Teachers' College and went \textcolor{blue}{missing} in \textcolor{magenta}{Iguala} on September 26, 2014, turned violent for the first time. Protesters set fire at the National Palace and some of them were arrested.}  \\ \hline
\multicolumn{4}{|c|}{Venezuela} \\ \hline
2014-02-18  &  \pbox{5cm}{represión(repression), disparó(shooting), \textcolor{blue}{marchamos(march)}, heridos(wounded), nicolasmaduro, armados(armed), \textcolor{red}{leopoldolopez}, \textcolor{orange}{apresar(arrest)}, \textcolor{green}{ntn24}}   &  Ultimas Notícias  & \pbox{7cm}{Mass \textcolor{blue}{protests} led by opposition leaders, including \textcolor{red}{Leopoldo López}, occurred in 38 cities across Venezuela asking for the release of the \textcolor{orange}{arrested students}. Colombian TV news channel NTN24 is taken off air for airing anti-government demonstrations.}   \\ \hline
\textcolor{orange}{2014-05-01}  &  \pbox{5cm}{muertes(deaths), cambio(change), \textcolor{red}{caracas}, presidente(president), \textcolor{orange}{labor}}    &   El Tiempo in Anzoategui  & \pbox{7cm}{Thousands of Venezuelans demonstrated in \textcolor{red}{Caracas} to commemorate \textcolor{orange}{Labor Day} and denounce shortages. Some protesters were injured when dispersed by authorities.}     \\ \hline
2014-08-12   &   \pbox{5cm}{gubernamental(government), \textcolor{red}{anticontrabando}, contrabando, \textcolor{green}{ébola}, muerte(death)}   &  El Nacional  &   \pbox{7cm}{Venezuela is the only country in Latin America with increasing number of \textcolor{red}{malaria}. With the spreading of \textcolor{green}{Ebola} virus, Venezuela is one of the most vulnerable countries in Latin America due to the lack of basic supplies, water, and electricity.}   \\ \hline
\end{tabular}
\end{table*}
\end{small}

\subsubsection{Case Studies}

\textbf{Mexico May 2014}. In December 2013, Mexican president Peña Nieto and Congress amended the Constitution, opening up the state controlled oil industry to foreign investors. Tens of thousands of protesters demonstrated in Mexico City on Labor Day (May 1) to protest against the energy reform, fearing the total privatization of the energy sector and higher energy prices \footnote{http://www.wbur.org/hereandnow/2014/05/02/may-day-mexico}
. In additions, protesters were also unsatisfied with the 2013 reforms of the educational sector. However, this incident was not reported in a number of influential newspapers in Mexico, including but not limited to Noroeste, Vanguardia, El Siglo de Torreon, Correo, El Imparcial, El Informador, Novedades Acapulco, and El Universal, which is an indicator of censorship. Fig. \ref{fig:casewc-a} shows a cluster of censored keywords detected by our method around May 1, 2014 in Mexico. 
Our approach has successfully captured consistent censorship patterns among a collection of relevant keywords, which well describe the topics around which the May 1 demonstrations were organized (reforma, gasolina, dinero, educación, escuela). \\ \indent

\textbf{Venezuela February 2014}. As a result of the collapse of the price of oil (main export of the country), a decade of disastrous macroeconomic policies and growing authoritarianism Venezuela suffered from inflation, shortages of basic foodstuffs and other necessities, and increasing political frustration. 
Mass opposition protests led by opposition leaders demanding the release of the students occurred in 38 cities across Venezuela on February 12, 2014. While this incident was reported by a number of major international newspapers,
there was significant censorship in the country's largest daily Ultimas Notícias, an event
reported by a number of international news outlets \footnote{http://articles.chicagotribune.com/2014-02-19/news/sns-rt-us-venezuela-protests-media-20140219\_1\_live-coverage-president-nicolas-maduro-news-channel-globovision/2} \footnote{http://www.nybooks.com/daily/2014/04/09/venezuela-protests-censorship/} \footnote{https://panampost.com/marcela-estrada/2014/02/13/venezuela-opposition-rallies-end-in-bloodshed-riots/}
. The day after the protests President Maduro announced that Colombian TV news channel NTN24, which had been the only station to broadcast the protests to within Venezuela, was being removed from the grid of Venezuelan cable operators for airing anti-government demonstrations. Fig. \ref{fig:casewc-b} shows a cluster of censored keywords detected by our method around February 18, 2014 in Venezuela, which well describes the populations involved (estudiante, chavistas, opositores, leopoldolopez) and the target of the demonstrations (nicolasmaduro). \\ \indent

\begin{figure}[!t]
\small
    \centering
    \begin{subfigure}[b]{0.22\textwidth}
            \centering
            \includegraphics[width=\textwidth]{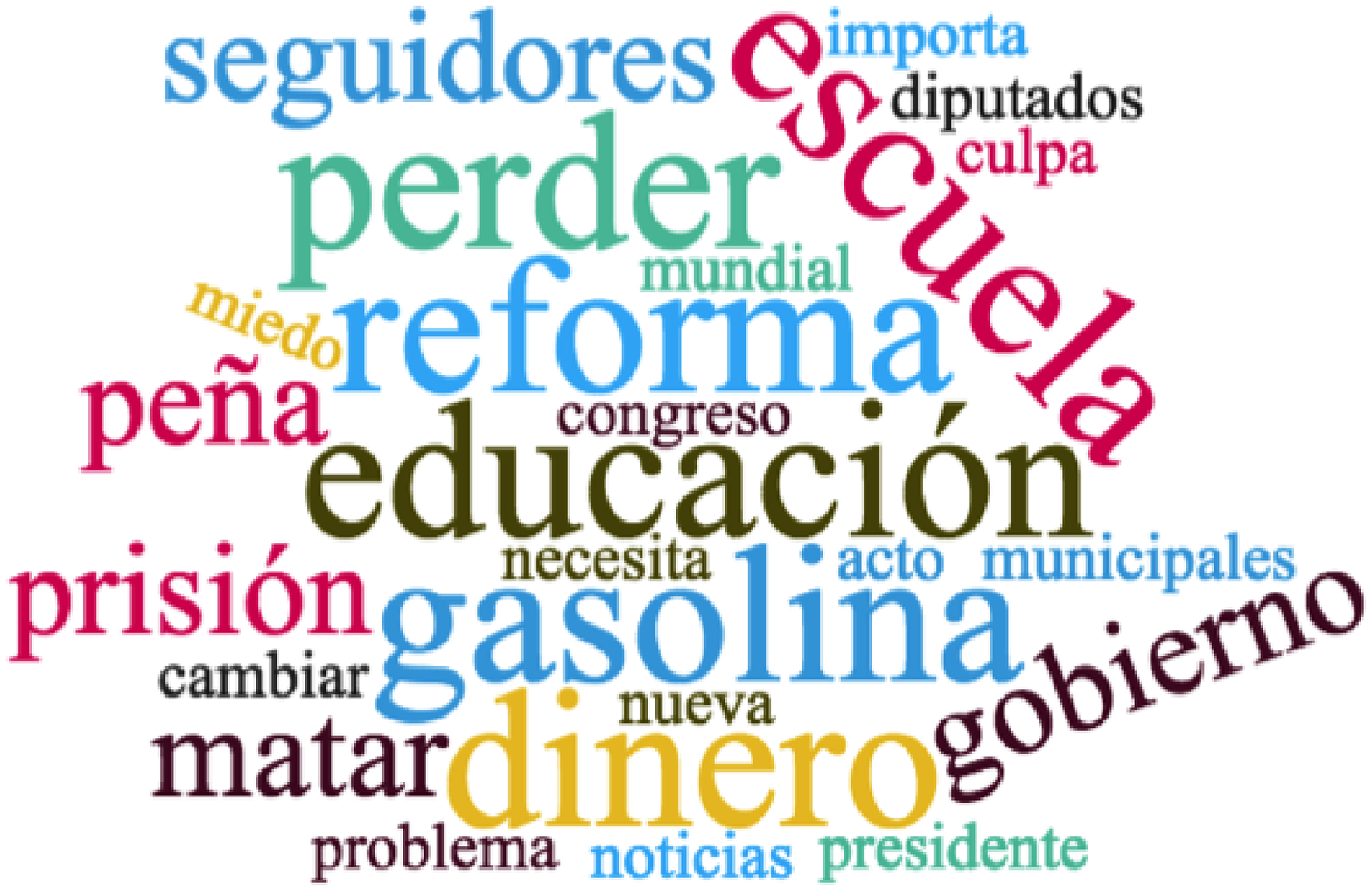}
            \caption{\scriptsize Mexico 2014-05-01}
    \label{fig:casewc-a}
    \end{subfigure}
    \begin{subfigure}[b]{0.22\textwidth}
            \centering
            \includegraphics[width=\textwidth]{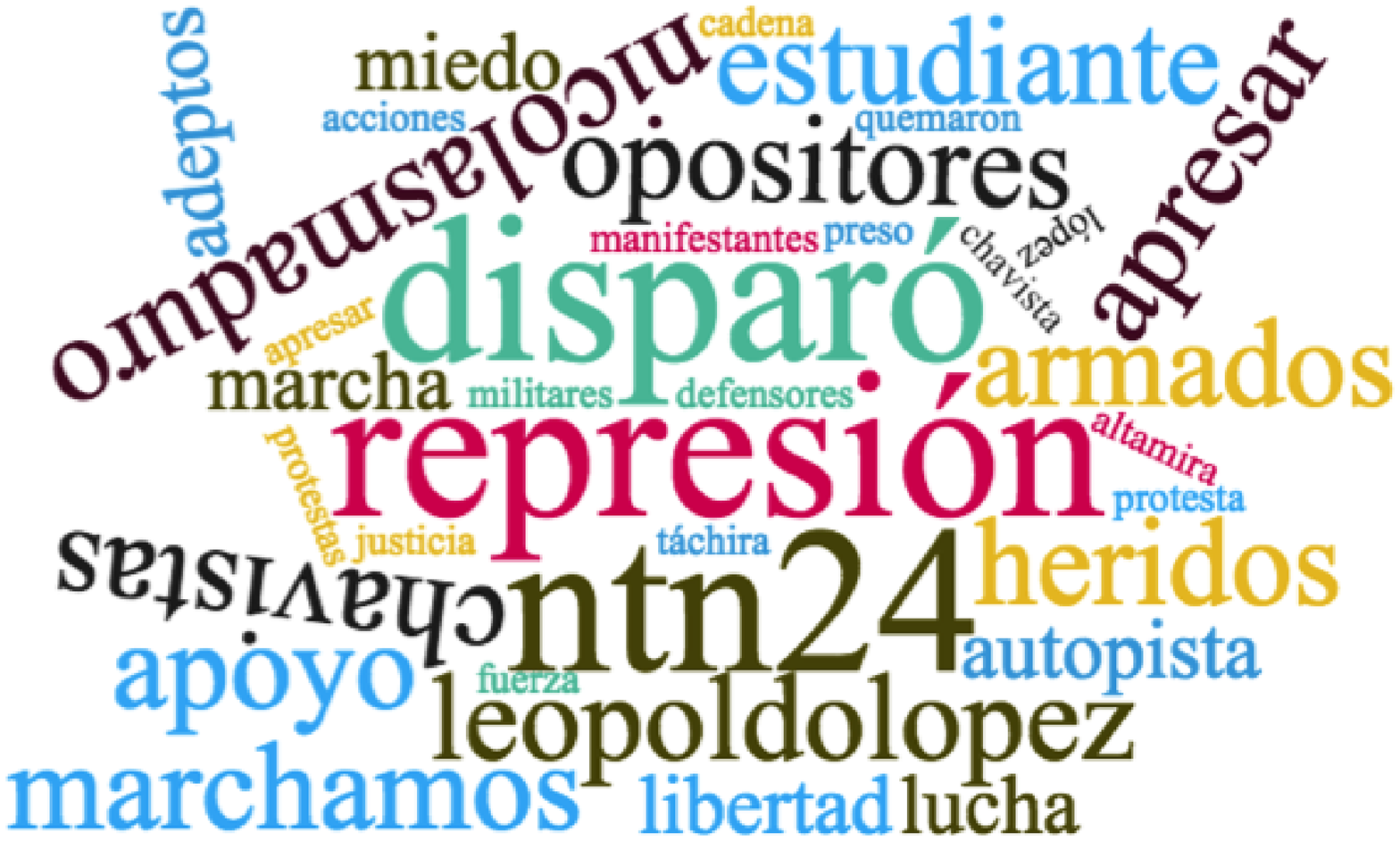}
    \caption{\scriptsize Venezuela 2014-02-18}
    \label{fig:casewc-b}
    \end{subfigure}    
    \caption{\scriptsize Word cloud representing censored keywords in \textit{News} identified by our method}
    \label{fig:casewc}
\end{figure}

\section{Conclusion}
In this paper, we have presented a novel unsupervised approach to identify censorship patterns in domestic news media using social media as a sensor. 
Through comprehensive experiments on semi-synthetic datasets, we showed that our approach outperforms popular anomalous subgraph detection methods: LTSS, EventTree, and NPHGS, especially when the perturbation intensity is low. Analyzing real-world datasets in Mexico and Venezuela during Year 2014 demonstrates that our approach is capable of accurately detecting real-world censorship patterns in domestic newspapers. In future work, we are interested in generalizing censorship detection to other countries and to undertake censorship forecasting.

\def\UrlBreaks{\do\/\do-}
\setlength{\bibsep}{0pt plus 0.3ex}


\end{document}